\documentclass[]{article}

\usepackage{authblk}
\usepackage{cite}
\usepackage{amsmath,amssymb,amsfonts,amsthm}
\usepackage{braket}
\usepackage{algorithm}
\usepackage[indLines=false,commentColor=teal]{algpseudocodex}
\usepackage{subcaption}
\usepackage{graphicx}
\usepackage{textcomp}
\usepackage{xcolor}
\usepackage{enumitem}
\usepackage{booktabs}
\usepackage[hidelinks]{hyperref}
\usepackage[margin=0.85in]{geometry}

\theoremstyle{plain}
\newtheorem{theorem}{Theorem}
\newtheorem{lemma}[theorem]{Lemma}

\newtheorem{corollary}[theorem]{Corollary}

\theoremstyle{definition}
\newtheorem{definition}{Definition}
\newtheorem{problem}{Problem}

\theoremstyle{remark}

\DeclareMathOperator{\wt}{wt}
\DeclareMathOperator{\supp}{supp}
\DeclareMathOperator{\syn}{syn}
\newcommand{\gap}{\mathsf{gap}}
\newcommand{\LER}{\mathrm{LER}}

\makeatletter %
\newcommand{\linebreakand}{%
\end{@IEEEauthorhalign}
\hfill\mbox{}\par
\mbox{}\hfill
\begin{@IEEEauthorhalign}
}
\makeatother %

\begin{document}

\title{Staying Sober: A Shot-Free Predictor of Logical Error Rates for Concatenated Quantum Error Correcting Codes}

\author[1]{Sayam Sethi\footnote{\href{mailto:sayams@utexas.edu}{sayams@utexas.edu}}}
\author[1]{Aditi Awasthi}
\author[1]{Maxwell Poster}
\author[2]{Joshua Viszlai}
\author[1]{Jonathan Mark Baker}
\affil[1]{Electrical and Computer Engineering\\The University of Texas at Austin}
\affil[2]{Department of Computer Science\\University of Chicago}

% \author{
%   \IEEEauthorblockN{Sayam Sethi\IEEEauthorrefmark{1}}
%   \IEEEauthorblockA{Electrical and Computer Engineering\\The
%   University of Texas at Austin}
%   \and
%   \IEEEauthorblockN{Aditi Awasthi}
%   \IEEEauthorblockA{Electrical and Computer Engineering\\The
%   University of Texas at Austin}
%   \and
%   \IEEEauthorblockN{Maxwell Poster}
%   \IEEEauthorblockA{Electrical and Computer Engineering\\The
%   University of Texas at Austin}
%   \and
%   \linebreakand{}
%   \IEEEauthorblockN{Joshua Viszlai}
%   \IEEEauthorblockA{Department of Computer Science\\University of Chicago}
%   \and
%   \IEEEauthorblockN{Jonathan Mark Baker}
%   \IEEEauthorblockA{Electrical and Computer Engineering\\The
%   University of Texas at Austin}
%   \linebreakand{}
% }

\maketitle

\begin{abstract}
% Estimating the Logical Error Rate (LER) of Quantum Error Correcting (QEC) codes is an important step in the design of Fault-Tolerant Quantum Computing (FTQC) systems. However, the problem of estimating LERs by Monte Carlo (MC) sampling scales in the number of samples as $O(p^{-d/2}/\varepsilon^2)$ in the physical error rate $p$, the code distance $d$ of the QEC code, and the statistical error of the sampling $\varepsilon$. As we discover newer QEC codes at higher distances with interesting properties, and as device error rates improve, the number of MC samples needed for a reliable estimate of LER becomes intractable. This problem is further exacerbated with the increased decoding time required per sample for codes with larger distances, and higher encoding rates. Prior work has addressed this issue, however their approach works for hard decoding and therefore does not extend to concatenated codes, which require soft-information decoding. We propose a novel LER prediction scheme, which takes in as input the stabilizer matrix $\mathcal{S}$ and the set of minimum weight logicals $\mathcal{L}$ of the QEC code, and outputs the leading order monomial of the LER in time that is $O(\text{poly}(d, |\mathcal{S}|, |\mathcal{L}|))$ under heuristic budgets. We prove the optimality of our approach under unconstrained budgets. Our scheme works for both code capacity and phenomenological noise, and supports a depolarizing noise models with arbitrary biases.
Estimating the logical error rate (LER) of a quantum error-correcting (QEC) code by Monte Carlo sampling takes $O(p^{-\lceil d/2 \rceil}/\varepsilon^2)$ samples at physical error rate $p$, code distance $d$ and relative error $\varepsilon$, which becomes intractable at large distances and low error rates, and especially for concatenated codes, whose distance grows as $d^{\ell}$ with the number of levels $\ell$. The fail fast method~\cite{beverland_fail_2025} reduces this cost by sampling the failures of each weight separately, but it only works for unconcatenated codes, and does not model the soft information that the levels of a concatenated code pass to each other. We propose a predictor that takes the stabilizers, the minimum-weight logical operators and, for odd $d$, the logical operators of weight $d + 1$ of a code, and computes, without sampling, the leading monomial $c\,p^{o}$ of the LER when expressed as a polynomial in $p$, at every level of concatenation, in time polynomial in the code parameters and linear in $\ell$ for fixed budgets. We prove that our prediction is exact for any unconcatenated stabilizer code, for every logical observable and arbitrary per-qubit error rates, whenever no budget is exhausted. Furthermore, for concatenated codes, we construct a recursive approach that computes the leading monomial exactly in exponential time, and our predictor approximates this approach. Our predictions match exact enumeration at level 1 on every code we evaluate, and up to three levels for the iceberg code. Using our predictor, we evaluate and compare $37$ code configurations under four noise biases in $48$ core-hours, where Monte Carlo sampling would need more than $10^{12}$ shots at $p = 10^{-3}$ for $113$ of these $148$ cases.
\end{abstract}

\section{Introduction}

\begin{figure}
    \centering
    \begin{subfigure}[t]{0.3\linewidth}
    \centering
    \includegraphics[width=\linewidth]{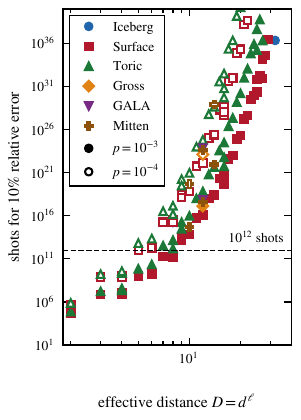}
    \caption{Monte Carlo shots for $10\%$ error.}
    \label{fig:cost}
    \end{subfigure}
    \begin{subfigure}[t]{0.69\linewidth}
        \centering
        \includegraphics[width=\linewidth]{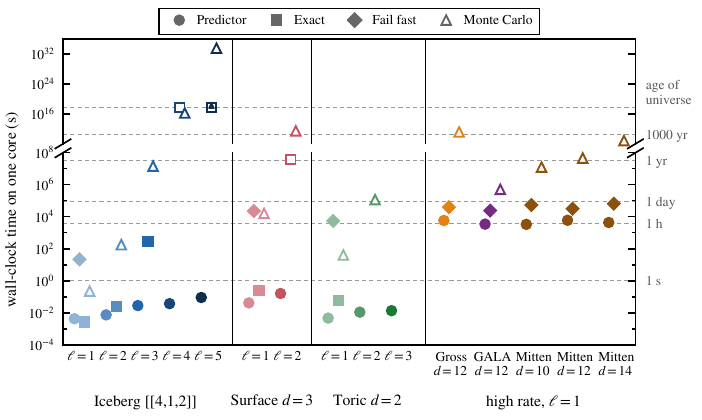}
        \caption{Wall-clock time of each method on one core.}
        \label{fig:runtime}
    \end{subfigure}
    \caption{Cost of estimating the LER under code-capacity noise. (a) Shots that Monte Carlo sampling needs to estimate the LER $L$ of the first logical qubit to a relative error of $10\%$, $(1 - L)/(0.01\,L)$, with $L$ from our prediction, at $p = 10^{-3}$ (filled) and $p = 10^{-4}$ (open). The dashed line marks $10^{12}$ shots, beyond which Monte Carlo sampling becomes impractical. (b) Wall-clock time of each method at $p = 10^{-3}$, where filled markers are measured, open markers are projected, and the arrow marks a lower bound. Our predictor gives the LER at every $p$ from one run, and so does the spectrum that fail fast samples, but at a precision that varies with $p$, while exact enumeration and Monte Carlo sampling must be rerun at each $p$. For the high-rate codes, Monte Carlo sampling and fail fast decode with BP-OSD, and every predictor time includes finding the minimum-weight logical operators.}
\end{figure}

Quantum computers promise speedups on scientifically important problems, including factoring, optimization, and quantum simulation. Their qubits are highly susceptible to errors, so quantum error correction (QEC) must protect the computation, by encoding information from many noisy physical qubits into fewer, more reliable logical qubits. A program then succeeds only if the logical error rate (LER) of the chosen code is low enough, often below $10^{-12}$.

Recent works have proposed many candidate codes for fault-tolerant architectures, from spatially-efficient quantum low-density parity-check (qLDPC) codes such as the gross code~\cite{bravyi_high-threshold_2024}, to concatenated codes. Concatenated codes in particular have seen rising interest lately, from constant-overhead fault tolerance and high-rate code families to compilers, architectures and experiments~\cite{yamasaki_time-efficient_2024, goto_many-hypercube_2024, yoshida_concatenate_2025, litinski_blocklet_2025, liu_coniq_2025, putterman_hardware-efficient_2025, sethi_exploring_2026}. Concatenating a code of distance $d$ with itself over $\ell$ levels~\cite{knill_concatenated_1996} gives a code of distance at least $d^{\ell}$, so the LER of the concatenated code falls quickly with each level.

Choosing a code therefore requires comparing the LERs of many candidates, which differ in their family, distance, number of logical qubits and number of levels, each under several noise models. Monte Carlo sampling estimates the LER by decoding sampled error configurations. Since the LER scales as $p^{\lceil d/2 \rceil}$ at physical error rate $p$, Monte Carlo sampling needs $O(p^{-\lceil d/2 \rceil}/\varepsilon^2)$ samples for a relative error $\varepsilon$. For a concatenated code, $d$ becomes $d^{\ell}$, so the number of samples grows exponentially in $d^{\ell}$. Figure~\ref{fig:cost} shows that $62$ of the $80$ codes we evaluate need more than $10^{12}$ samples at $p = 10^{-3}$, and Figure~\ref{fig:runtime} shows projected Monte Carlo run times beyond the age of the universe. Moreover, each Monte Carlo run gives the LER at a single $p$, so every new $p$ needs a new run.

Prior approaches either still sample or do not apply to these codes: 1. rare-event methods such as splitting~\cite{bravyi_simulation_2013} and fail fast~\cite{beverland_fail_2025} need fewer samples but still decode every sample, 2. analytic methods give closed forms only for specific codes~\cite{fowler_analytic_2012, regev_closed_2026, forlivesi_logical_2023}, or, through weight enumerators~\cite{forlivesi_performance_2023, cao_quantum_2023}, the LER of bounded-distance decoding or the probability of an undetected logical error, rather than the LER of the maximum-likelihood (ML) decoder, 3. analytic methods for concatenated codes assume hard decoding or stop at two levels~\cite{rahn_exact_2001, fern_correctable_2008, zhang_bidirectional_2026}, and 4. soft decoding approaches for concatenated codes~\cite{poulin_optimal_2006, evans_message_2008, gidney_yoked_2023, meister_efficient_2024} estimate the LER by sampling. Sampling also needs an efficient decoder, which for qLDPC codes is typically BP-OSD~\cite{roffe_decoding_2020}. BP-OSD fails on some error configurations of weight below $d/2$, so fail fast with BP-OSD estimates LERs up to $2.5 \times 10^{9}$ times higher than our prediction at $p = 10^{-3}$ (Figures~\ref{fig:failfast} and~\ref{fig:failfast-mitten}). Other works rank codes by proxies such as an effective distance or the number of minimum-weight logical operators~\cite{dua_clifford-deformed_2024, perlin_quantum_2026}, which do not give the LER.

In this work, we propose a predictor that, to our knowledge, is the first to compute, without sampling, the LER of the soft tree decoder of concatenated codes beyond two levels. It takes the stabilizers, the minimum-weight logical operators and, for odd $d$, the logical operators of weight $d + 1$ of a code, and computes the leading term $c\,p^{o}$ of the LER at every level, which gives the LER at every sufficiently small $p$ from one run. At small $p$, the decoder compares the leading orders of the likelihoods of the logical classes, so whether it fails on an error configuration depends only on the gaps between these orders. We count each failing error configuration at the intersection of the minimum-weight logical operators containing it, and use a recursion that computes the distribution of these gaps at each level from that of the level below. Budgets on this recursion make the predictor run in polynomial time. Our predictions agree with exact enumeration to within $0.1\%$ at level 1 and $p = 3.2 \times 10^{-4}$, and we rank $37$ code configurations in $48$ core-hours, where Monte Carlo sampling would need more than $10^{12}$ shots at $p = 10^{-3}$ for $113$ of the $148$ cases we compare. Our contributions are as follows:
\begin{enumerate}[nosep]
    \item We propose a predictor that computes, without sampling, the leading term $c\,p^{o}$ of the LER of a code at every level of concatenation, and prove that the predictor is exact for any unconcatenated stabilizer code, for every logical observable and arbitrary per-qubit error rates, whenever no budget is exhausted (Theorem~\ref{thm:exact}).
    \item We construct a recursive approach for computing the leading term of the LER for concatenated codes. We then approximate the recursive approach within budgets, in time polynomial in the code parameters and linear in the number of levels.
    \item We validate our predictor against exact enumeration, which it matches to within $0.1\%$ at level 1 and $p = 3.2 \times 10^{-4}$, and to within $0.2\%$ for the iceberg code up to level 3.
    \item We rank $37$ code configurations under four noise biases and two noise models in $48$ core-hours, and identify the best code among the contenders at the same effective distance. We observe that the choice of best code depends on the noise bias.
\end{enumerate}

The rest of this paper is organized as follows. Section~\ref{sec:background} introduces stabilizer codes, ML decoding and concatenation, and surveys prior work. Section~\ref{sec:computing} computes the exact leading-order LER of an unconcatenated code, and Section~\ref{sec:concatenated} extends it to concatenated codes. Section~\ref{sec:approximation} approximates this computation within budgets, Section~\ref{sec:evaluation} evaluates our predictor, and Section~\ref{sec:conclusion} concludes.
\section{Background}\label{sec:background}

\subsection{Stabilizer Codes}
Let $\mathcal{P}_1 = \{I, X, Y, Z\}$ denote the single-qubit Pauli matrices.
Let $\mathcal{P}_n$ be the $n$-qubit Pauli group, consisting of operators $P = \gamma P_1 \otimes \dots \otimes P_n$, where $P_i \in \mathcal{P}_1$ and the phase is $\gamma \in \{\pm 1, \pm i\}$.
The support of $P$ is the set $\supp(P) = \{i \in [n] : P_i \neq I\}$ of qubits on which it acts non-trivially, and its weight is $\wt(P) = |\supp(P)|$.
For Pauli operators $A$ and $B$, we write $A \subseteq B$ when $A$ acts as $B$ on every qubit of $\supp(A)$, which is stronger than $\supp(A) \subseteq \supp(B)$.
%We typically ignore global phases and identify $P$ as $\pm P$ and $\pm i P$

For a given $[[n, k, d]]$ QEC code, the stabilizer group $S = \langle g_1, \dots, g_{n-k} \rangle \subset \mathcal{P}_n$ is an abelian subgroup that does not contain $-I$, specified by $n - k$ independent generators.
The code space is the joint $+1$ eigenspace of $S$, i.e., the states $\ket{\psi}$ with $g \ket{\psi} = \ket{\psi}$ for every $g \in S$.
This encodes $k$ logical qubits into $n$ physical qubits.
A Pauli error $E \in \mathcal{P}_n$ is diagnosed by measuring each generator.
We collect the outcomes in the syndrome,
\begin{equation}
\syn(E)_i = \begin{cases} 0 & \text{if } g_i E = E g_i, \\ 1 & \text{if } g_i E = -E g_i. \end{cases}
\end{equation}
Crucially, $\syn(E \cdot E') = \syn(E) \oplus \syn(E')$, i.e., $\syn$ is a homomorphism.

The kernel of $\syn$ is the normalizer of $S$ in $\mathcal{P}_n$,
\begin{equation}
N(S) = \{P \in \mathcal{P}_n : \syn(P) = 0\} \supseteq S.
\end{equation}
Logical operators are the elements of $N(S)$.
Those that are not stabilizers, up to phase, are non-trivial, and act non-trivially on the encoded information.
Two logical operators that differ by a stabilizer act identically, so distinct logical actions correspond to the cosets in $N(S)/S$, which we call logical classes.
The $4^k$ logical classes form a group, with the product $(A S) \cdot (B S) = (A \cdot B) S$.
The errors with a given syndrome split into cosets of $S$ in the same way.
To name them, we fix for each syndrome one error with that syndrome as its reference, with the identity as the reference of the trivial syndrome.
The logical class $\mathsf{cls}(E)$ of an error $E$ is then the coset in $N(S)/S$ of $E$ times the reference of $\syn(E)$.
Two errors with the same syndrome lie in the same class exactly when they differ by a stabilizer, up to phase.
For two such errors $E$ and $E'$, we also have $\mathsf{cls}(E \cdot E') = \mathsf{cls}(E) \cdot \mathsf{cls}(E')$.
A phase does not change the class, and $\lambda \cdot \lambda$ is the trivial class $S$ for every class $\lambda$.

The usual choice for logical observables is $\bar X_j$ and $\bar Z_j$ for each of the $k$ logical qubits, $j = 1, \dots, k$, satisfying the canonical commutation relation at the logical level, i.e., $\bar X_j \bar Z_j = -\bar Z_j \bar X_j$, while the observables of different logical qubits commute.
We write $\theta$ for a logical observable.
A logical class flips $\theta$ if its elements anti-commute with $\theta$, which does not depend on the chosen element, since $\theta$ commutes with every stabilizer.
Thus, $\lambda \cdot \lambda'$ flips $\theta$ exactly when exactly one of $\lambda$ and $\lambda'$ flips it.
The distance $d$ of the code is the minimum weight of a non-trivial logical operator.

\subsection{Noise and Maximum-Likelihood Decoding}\label{sec:ml-decoding}

We model noise in the \emph{code-capacity} setting, in which each data qubit independently suffers a Pauli error and syndrome measurements are assumed to be perfect.
Qubit $i$ experiences $X$, $Y$, or $Z$ with probabilities $b_{i,X}\, p$, $b_{i,Y}\, p$, and $b_{i,Z}\, p$, respectively, and is left untouched with probability $1 - p$, where $p$ is the physical error rate and the biases $b_{i,P} > 0$ satisfy $b_{i,X} + b_{i,Y} + b_{i,Z} = 1$.
\emph{Depolarizing} noise is the symmetric case $b_{i,X} = b_{i,Y} = b_{i,Z} = 1/3$.
The errors on the $n$ data qubits form an \emph{error configuration} $E = E_1 \otimes \dots \otimes E_n$ with each $E_i \in \mathcal{P}_1$, which occurs with probability
\begin{equation}\label{eq:configuration-prob}
\Pr(E) = b_E\, p^{\wt(E)} (1 - p)^{n - \wt(E)}, \qquad b_E = \prod_{i \in \supp(E)} b_{i,E_i},
\end{equation}
where we call $b_E$ the bias product of $E$.
Error configurations carry no phase, and we denote the set of all $4^n$ of them by $\mathcal{P}_1^{\otimes n}$.

\emph{Biased} noise has the same biases on every qubit, with one Pauli more likely than the other two.
In this work, noise biased toward $P \in \{X, Y, Z\}$ has $b_{i,P} = 0.8$ and a bias of $0.1$ for each of the other two Paulis, so the dominant Pauli is $8\times$ more likely than either of the others.

Given an error configuration $E$, a \emph{decoder} maps the observed syndrome $\syn(E)$ to a recovery operator $R$ with $\syn(R) = \syn(E)$.
Since $\syn$ is a homomorphism, the residual $R \cdot E$ has trivial syndrome, so $R \cdot E \in N(S)$.
Decoding succeeds if $R \cdot E$ is a stabilizer, up to phase, and fails otherwise, i.e., if the residual is a non-trivial logical operator.
Only the class of $R$ matters: writing $\hat\lambda = \mathsf{cls}(R)$ for the class that the decoder returns, the residual has class $\hat\lambda \cdot \mathsf{cls}(E)$.
Decoding therefore fails exactly when $\hat\lambda \neq \mathsf{cls}(E)$, and it fails for a logical observable $\theta$ when $\hat\lambda \cdot \mathsf{cls}(E)$ flips $\theta$.
The decoder only sees the syndrome, so it compares the following likelihoods of the classes at that syndrome.

\begin{definition}[Class likelihood]\label{def:class-likelihood}
For a syndrome $s$ and a logical class $\lambda$, the class likelihood $W_\lambda(s)$ is the total probability $\Pr(E)$ of the error configurations $E$ with $\syn(E) = s$ and $\mathsf{cls}(E) = \lambda$.
\end{definition}

The maximum-likelihood (ML) decoder returns the class with the largest $W_\lambda(s)$, which minimizes the probability that decoding fails at each syndrome.
When several classes share the largest likelihood, we call them \emph{ambiguous}, and the decoder returns each of them with equal probability.
The \emph{logical error rate} (LER) is the probability that decoding fails, which we also call the total LER.

The LER can also be resolved by logical observable: $\LER_\theta(p)$ is the probability that decoding fails for $\theta$.
We report the LER of the first logical qubit, i.e., the probability that decoding fails for $\bar X_1$ or $\bar Z_1$, which is the total LER when $k = 1$.

\begin{figure}
\centering
\begin{tikzpicture}[
  block/.style={draw, rectangle, minimum width=22mm, minimum height=6.5mm, inner sep=2pt, font=\small},
  q1/.style={draw=blue!75!black, line width=0.6pt},
  q2/.style={draw=red!75!black, line width=0.6pt, dashed},
  qk/.style={draw=green!45!black, line width=0.6pt, dash dot},
  port/.style={circle, inner sep=0pt, minimum size=3.6pt},
  fan/.style={line width=0.4pt},
  level/.style={font=\small, anchor=west, align=left},
]
  % the diagram is stretched to fill the text width; the labels sit outside it
  \begin{scope}[xscale=1.05]
  % blocks at level l - 1; logical qubits 1, 2 and k_1 sit on the top edge
  \foreach \i/\x/\lab in {1/0/1, 2/3.0/2, 3/10.6/n_2} {
    \node[block] (in\i) at (\x, 0) {block $\lab$};
    \foreach \j/\f in {1/0.2, 2/0.4, 3/0.8}
      \coordinate (p\i\j) at ($(in\i.north west)!\f!(in\i.north east)$);
    \node[font=\scriptsize, inner sep=0pt, above=1.5pt] at ($(in\i.north west)!0.6!(in\i.north east)$) {$\cdots$};
    \foreach \dx in {-0.66, -0.22, 0.22, 0.66}
      \draw[fan] (in\i.south) -- ++(\dx, -0.65) node[anchor=north, inner sep=0pt, yshift=1pt] {$\vdots$};
  }
  \node at (6.8, 0) {$\cdots$};
  % blocks at level l; qubits 1, 2 and n_2 sit on the bottom edge
  \foreach \j/\x/\lab in {1/2.0/1, 2/4.7/2, 3/8.6/k_1} {
    \node[block] (out\j) at (\x, 3) {block $\lab$};
    \node[font=\scriptsize, inner sep=0pt, below=1.5pt] at ($(out\j.south west)!0.6!(out\j.south east)$) {$\cdots$};
  }
  \node at (6.65, 3) {$\cdots$};
  % the j-th level-l block takes the j-th logical qubit of every level-(l-1) block
  \foreach \i/\f in {1/0.2, 2/0.4, 3/0.8} {
    \draw[q1] (p\i1) -- ($(out1.south west)!\f!(out1.south east)$);
    \draw[q2] (p\i2) -- ($(out2.south west)!\f!(out2.south east)$);
    \draw[qk] (p\i3) -- ($(out3.south west)!\f!(out3.south east)$);
    \node[port, fill=blue!75!black] at (p\i1) {};
    \node[port, fill=red!75!black] at (p\i2) {};
    \node[port, fill=green!45!black] at (p\i3) {};
  }
  % legend, one row under the figure
  \foreach \s/\x/\lab in {q1/-0.2/1, q2/3.7/2, qk/7.6/k_1} {
    \draw[\s] (\x, -2.05) -- ++(0.6, 0);
    \node[right, font=\small] at ({\x + 0.65}, -2.05) {logical qubit $\lab$};
  }
  \end{scope}
  % level labels, with the code each level forms on the physical qubits
  \node[level] at (-4.85, 3) {level $\ell$ \\ $[[n_1 n_2, k_1 k_2, \geq d_1 d_2]]$};
  \node[level] at (-4.85, 0) {level $\ell - 1$ \\ $[[n_1, k_1, d_1]]$};
\end{tikzpicture}
\caption{Interleaved concatenation. Each block at level $\ell - 1$ encodes $k_1$ logical qubits in an $[[n_1, k_1, d_1]]$ code. Each block at level $\ell$ uses $n_2$ logical qubits from the lower level as the physical qubits of an $[[n_2, k_2, d_2]]$ code. The $j\textsuperscript{th}$ block at level $\ell$ takes the $j\textsuperscript{th}$ logical qubit of every block at level $\ell - 1$. Thus, no block at level $\ell$ contains two logical qubits of the same block at level $\ell - 1$. Together, the $k_1$ blocks at level $\ell$ form an $[[n_1 n_2, k_1 k_2, \geq d_1 d_2]]$ code. The edges below each block at level $\ell - 1$ continue its tree down to the physical qubits.}
\label{fig:interleaved}
\end{figure}
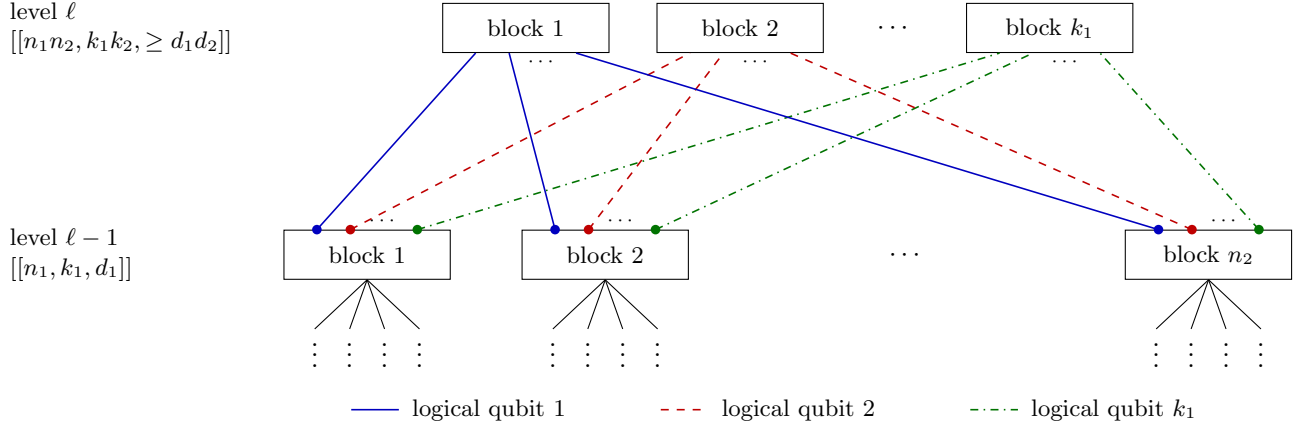

\subsection{Concatenation}\label{sec:tree-decoder}
Code concatenation is a systematic way to construct larger QEC codes from smaller codes, raising the effective distance without needing to design new codes with the desired $[[n,k,d]]$ parameters. Given an inner $[[n_1, k_1, d_1]]$ code and an outer $[[n_2, k_2, d_2]]$ code, each logical qubit of the inner code serves as a physical qubit of the outer code. This induces a natural tree representation. The leaves are physical qubits, and each internal node at level $\ell$ is a single level-$\ell$ block. The children of a node are the level-$(\ell-1)$ blocks whose logical qubits it encodes, and the root carries the encoded logical information.

Exact ML decoding of a general stabilizer code is \#P-hard~\cite{iyer_hardnessdecoding_2013} and, by direct enumeration, requires time exponential in the block size. The tree structure, however, permits decoding level by level, from the leaves to the root~\cite{poulin_optimal_2006}. For each child block $j$, the decoder computes the class likelihoods $W^{(j)}_\lambda(s^{(j)})$ of Definition~\ref{def:class-likelihood} at the syndrome $s^{(j)}$ of that block, which are proportional to the posterior distribution over the logical classes of that block. Rather than committing to a single correction, it passes this \emph{soft information} to the parent. For a block encoding a single logical qubit, this message consists of four values, one per logical class. At the leaves, a physical qubit $i$ passes up $1 - p$ for $I$ and $b_{i,P}\, p$ for each $P \in \{X, Y, Z\}$.

At the next level, the child logical qubits effectively act as physical qubits for the parent block. Since the noise on different child blocks is independent, the parent weighs an error $P = P_1 \otimes \dots \otimes P_n$ on its qubits by the product of the likelihoods that the child blocks give to the classes of $\bar{P_1}, \dots, \bar{P_n}$, where $\bar{P_j}$ is the logical operator of child block $j$ that acts as $P_j$ on its logical qubit. Given the parent syndrome $s$, the decoder sums these weights over all errors consistent with $s$ and a given parent logical class $\lambda$,
\begin{equation}\label{eq:tree-likelihood}
W_\lambda(s) = \sum_{\substack{P \in \mathcal{P}_1^{\otimes n} \\ \syn(P) = s,\ \mathsf{cls}(P) = \lambda}} \; \prod_{j=1}^{n} W^{(j)}_{\mathsf{cls}(\bar{P_j})}\left( s^{(j)} \right).
\end{equation}
The parent block passes these class likelihoods on to its own parent, and at the root, the decoder returns the class with the largest one.

\begin{theorem}[Tree decoding is ML decoding~\cite{poulin_optimal_2006}]\label{thm:tree-ml}
If every block encodes a single logical qubit, then the class likelihoods of Equation~\eqref{eq:tree-likelihood} at the root are the class likelihoods of the whole concatenated code, so the tree decoder is its ML decoder.
\end{theorem}

 For blocks with $k > 1$, we use interleaved concatenation~\cite{litinski_blocklet_2025}, shown in Figure~\ref{fig:interleaved}, so that no block contains two logical qubits of the same child block. The decoder we use passes the likelihoods of each logical qubit of a child block separately, so Theorem~\ref{thm:tree-ml} does not apply.

\begin{figure}
    \centering
    \includegraphics[width=\linewidth]{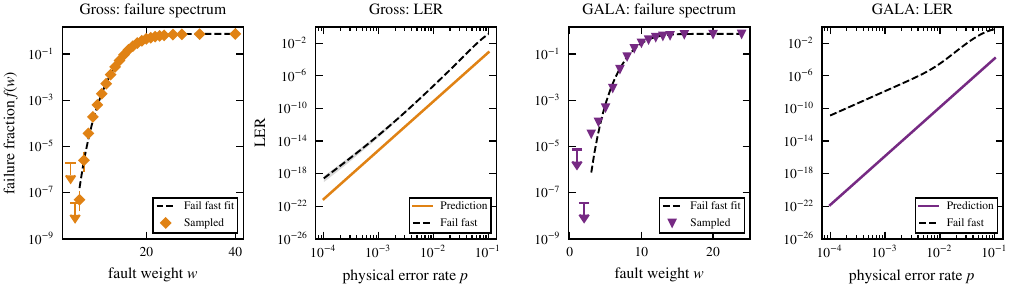}
    \caption{Fail fast with BP-OSD against our prediction on the gross and GALA codes at level 1. For each code, the failure spectrum sampled with BP-OSD with the fail fast fit to it, and the LER that fail fast obtains from the spectrum against our prediction.}
    \label{fig:failfast}
\end{figure}

\subsection{Prior Work}
Fail fast~\cite{beverland_fail_2025} samples the fraction of error configurations of each weight on which the decoder fails, and identifies the failing error configurations of the leading weight through the logical operators of weight $d$, and $d + 1$ for odd $d$, motivating the approach used in Theorem~\ref{thm:exact}. With BP-OSD, Figures~\ref{fig:failfast} and~\ref{fig:failfast-mitten} show that fail fast estimates LERs between $69\times$ and $2.5 \times 10^{9}\times$ higher than our prediction at $p = 10^{-3}$. For every code in these figures, the sampled failure spectrum starts below weight $d/2$, so the estimate of fail fast has a lower leading order than our prediction. For the gross code, BP-OSD fails on only $5$ of $10^{8}$ error configurations of weight $5$, so this lower order dominates only below $p \approx 10^{-3}$.

\begin{figure}
    \centering
    \includegraphics[width=\linewidth]{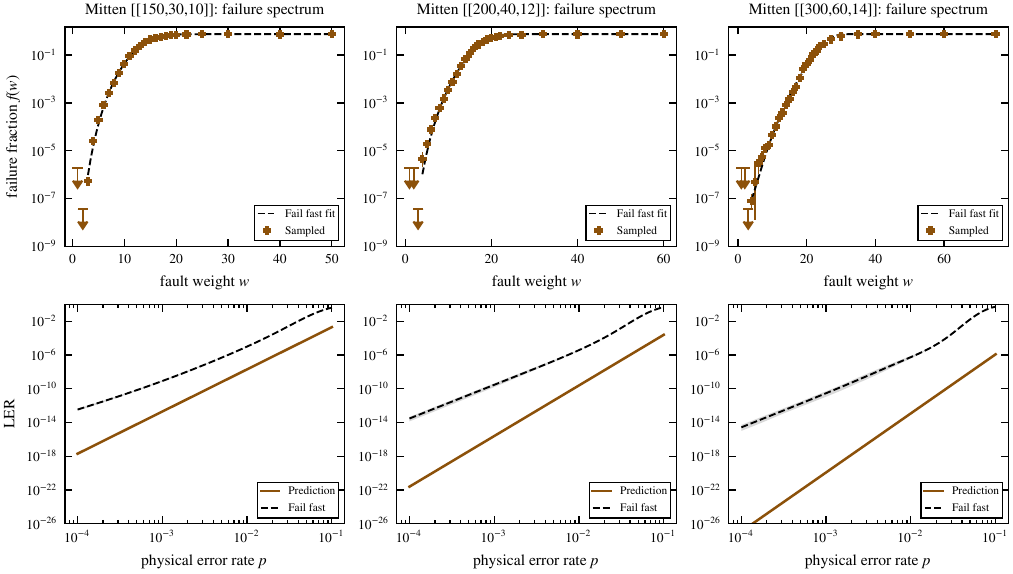}
    \caption{Fail fast with BP-OSD against our prediction on the Mitten codes at level 1, as in Figure~\ref{fig:failfast}. Top: the failure spectrum sampled with BP-OSD. Bottom: the LER that fail fast obtains from it against our prediction.}
    \label{fig:failfast-mitten}
\end{figure}
\section{Exact Leading-Order LER for Unconcatenated Codes}\label{sec:computing}
Our main result for an unconcatenated code, Theorem~\ref{thm:exact}, gives the leading term of the LER exactly. For a logical observable $\theta$ that some logical operator of weight $d$ flips, $\LER_\theta(p) / p^{\lceil d/2 \rceil}$ tends to a coefficient $c_\theta$ as $p \to 0$. This coefficient sums the bias products of the error configurations of weight $\lceil d/2 \rceil$ whose support lies inside that of a logical operator of weight $d$ or $d + 1$, each multiplied by the probability that the decoder fails for $\theta$ on that error configuration. Computing $c_\theta$ needs only these logical operators and the stabilizers of weight at most $d + 1$, without sampling, and $c_\theta\, p^{\lceil d/2 \rceil}$ then approximates $\LER_\theta(p)$ at every small $p$ at once. The result holds for any number of logical qubits and any positive biases.

A soft-information decoder outputs not only the most likely logical class but also its confidence, i.e., the complementary gap~\cite{gidney_yoked_2023} between the log-likelihoods of the two most likely classes, which depends only on the syndrome. As $p \to 0$, each class likelihood is asymptotically a single power of $p$, so this gap, measured from the true class in units of $\log(1/p)$, is an integer: the decoder fails when the gap is negative, and when the gap is zero the leading coefficients decide. We can therefore count error configurations by their gaps, with dynamic programming (DP) over accumulated gaps, instead of sampling them or enumerating all $4^n$ of them. Section~\ref{sec:leading-order} reduces the decoder's decision to comparing these powers and coefficients, Section~\ref{sec:failing} identifies the error configurations on which it fails and arrives at Theorem~\ref{thm:exact}, and Section~\ref{sec:counting} counts each of them once. The proofs of the lemmas and of Corollary~\ref{cor:gap-fraction} are in Appendix~\ref{app:unconcatenated}.

\subsection{Soft-Information Decoding Computes Gaps}\label{sec:leading-order}
Grouping the error configurations in the class likelihood $W_\lambda(s)$ of Definition~\ref{def:class-likelihood} by weight, Equation~\eqref{eq:configuration-prob} gives
\begin{equation}\label{eq:weight-expansion}
W_\lambda(s) = \sum_{w=0}^{n} m_{\lambda,w}(s)\, p^{w} (1-p)^{n-w}, \qquad m_{\lambda,w}(s) = \sum_{\substack{\syn(E) = s,\ \mathsf{cls}(E) = \lambda \\ \wt(E) = w}} b_E.
\end{equation}
Writing $\omega_\lambda(s)$ for the smallest $w$ with $m_{\lambda,w}(s) > 0$, which exists because every error rate is positive, so every class contains error configurations at every syndrome, the term $w = \omega_\lambda(s)$ dominates the sum as $p \to 0$. We make this precise with notation that we use throughout for likelihoods, and for their sums, products and quotients that the decoder compares. For positive functions $f(p)$ and $g(p)$ on $0 < p < 1$, we write $f \sim g$, and say that $f$ is asymptotic to $g$, if $\lim_{p \to 0} f(p)/g(p) = 1$. We say that $f$ has leading term $q\, p^{o}$, for an integer $o$ and $q > 0$, if $f \sim q\, p^{o}$.

\begin{lemma}[Leading terms]\label{lem:leading-terms}
Let $f_a$ and $g_a$ be positive functions of $p$ on $0 < p < 1$ for $a = 1, \ldots, r$. Then
\begin{enumerate}
    \item[(i)] if $f_a \sim g_a$ for every $a$, then $\sum_a f_a \sim \sum_a g_a$, $\max_a f_a \sim \max_a g_a$, $\prod_a f_a \sim \prod_a g_a$ and $f_1 / f_2 \sim g_1 / g_2$;
    \item[(ii)] \sloppy for integers $o_a$ and $q_a > 0$, with $o = \min_a o_a$, $\sum_a q_a\, p^{o_a} \sim \left( \sum_{a \,:\, o_a = o} q_a \right) p^{o}$ and $\max_a q_a\, p^{o_a} \sim \left( \max_{a \,:\, o_a = o} q_a \right) p^{o}$;
    \item[(iii)] if $q\, p^{o} \sim q'\, p^{o'}$ with $q, q' > 0$, then $o = o'$ and $q = q'$.
\end{enumerate}
\end{lemma}

Since $(1-p)^{n-w} \to 1$, each term of Equation~\eqref{eq:weight-expansion} with $m_{\lambda,w}(s) > 0$ is asymptotic to $m_{\lambda,w}(s)\, p^{w}$, so Lemma~\ref{lem:leading-terms}~(i) and~(ii) give,
\begin{equation}\label{eq:leading-limit}
W_\lambda(s) \sim m_{\lambda,\omega_\lambda(s)}(s)\, p^{\omega_\lambda(s)}.
\end{equation}

\begin{definition}[Leading order and leading coefficient]\label{def:leading-order}
The leading order of class $\lambda$ at syndrome $s$ is $\omega_\lambda(s)$, and its leading coefficient is $m_\lambda(s) = m_{\lambda,\omega_\lambda(s)}(s)$, so that $W_\lambda(s) \sim m_\lambda(s)\, p^{\omega_\lambda(s)}$ by Equation~\eqref{eq:leading-limit}.
\end{definition}

\begin{lemma}[Leading-order decision]\label{lem:decision-rule}
Fix a syndrome $s$, and suppose that among the classes of smallest leading order a single class $\lambda^*$ has the largest leading coefficient. Then there is a $p_0 > 0$ such that for every $p < p_0$ the ML decoder returns $\lambda^*$.
\end{lemma}

The argument of Lemma~\ref{lem:decision-rule} applies to any two classes whose leading orders or leading coefficients differ: below some threshold, the class with the smaller leading order, or at equal order the larger leading coefficient, is strictly more likely. Since there are finitely many syndromes and pairs of classes, the smallest of these thresholds, $p_0^*$, is positive.

\begin{definition}[Sufficiently small error rate]\label{def:suff-small}
A physical error rate $p$ is \emph{sufficiently small} if $p < p_0^*$.
\end{definition}

When several classes share both the smallest leading order and the largest leading coefficient, Lemma~\ref{lem:decision-rule} does not decide between them, and we split the error configurations evenly among them. Any other way of resolving the ambiguity gives the same leading coefficient for the total LER, but for a single logical observable it can change the leading coefficient when three or more classes are ambiguous.

\begin{definition}[Signed gap and competing classes]\label{def:signed-gap}
For an error configuration $E$ with $\lambda_0 = \mathsf{cls}(E)$ and syndrome $s = \syn(E)$, we call the classes $\lambda \neq \lambda_0$ competing classes. For each competing class $\lambda$, its signed gap and its coefficient ratio at $E$ are,
\begin{equation}\label{eq:class-gap}
\gap_\lambda(E) = \omega_\lambda(s) - \omega_{\lambda_0}(s), \qquad \rho_\lambda(E) = \frac{m_\lambda(s)}{m_{\lambda_0}(s)}.
\end{equation}
Write $\omega_{\neg\lambda_0}(s) = \min_{\lambda \neq \lambda_0} \omega_\lambda(s)$, let $\Lambda'_E = \{\lambda \neq \lambda_0 : \omega_\lambda(s) = \omega_{\neg\lambda_0}(s)\}$ be the set of competing classes of smallest leading order, and let $\rho_{\neg\lambda_0}(E) = \max_{\lambda \in \Lambda'_E} \rho_\lambda(E)$. The signed gap of $E$ is $\gap(E) = \omega_{\neg\lambda_0}(s) - \omega_{\lambda_0}(s)$, i.e., the smallest $\gap_\lambda(E)$ over the competing classes.
\end{definition}

Dividing Equation~\eqref{eq:leading-limit} for a competing class $\lambda$ by the same equation for $\lambda_0$, Lemma~\ref{lem:leading-terms}~(i) gives the signed gap as a likelihood ratio,
\begin{equation}\label{eq:class-ratio}
\frac{W_\lambda(s)}{W_{\lambda_0}(s)} \sim \rho_\lambda(E)\, p^{\gap_\lambda(E)}.
\end{equation}
When $\gap(E) = 0$, the coefficient ratios decide, which is why our predictor also stores them.

\subsection{Identifying Failing Error Configurations}\label{sec:failing}
To predict the LER, we first write it in the same form as the class likelihood.

\begin{definition}[Failure fraction]\label{def:failure-fraction}
For a logical observable $\theta$ and an error configuration $E$ with $\lambda_0 = \mathsf{cls}(E)$, the failure fraction $f_\theta(E) \in [0, 1]$ is the probability that $\hat\lambda \cdot \lambda_0$ flips $\theta$, where $\hat\lambda$ is the class returned by the decoder on $E$.
\end{definition}

The signed gap gives the failure fraction directly.

\begin{corollary}[Failure fraction from the gap]\label{cor:gap-fraction}
Let $E$ be an error configuration with $\lambda_0 = \mathsf{cls}(E)$ and syndrome $s$, and let $\theta$ be a logical observable. Let $\Lambda_E = \{\lambda \in \Lambda'_E : \rho_\lambda(E) = \rho_{\neg\lambda_0}(E)\}$, and let $\Lambda_E^\theta \subseteq \Lambda_E$ be the subset for which $\lambda \cdot \lambda_0$ flips $\theta$. Then, for \emph{sufficiently small} $p$,
\begin{equation}\label{eq:gap-fraction}
f_\theta(E) = \begin{cases} 0 & \text{if } \gap(E) > 0, \text{ or } \gap(E) = 0 \text{ and } \rho_{\neg\lambda_0}(E) < 1, \\ |\Lambda_E^\theta| / (|\Lambda_E| + 1) & \text{if } \gap(E) = 0 \text{ and } \rho_{\neg\lambda_0}(E) = 1, \\ |\Lambda_E^\theta| / |\Lambda_E| & \text{if } \gap(E) < 0, \text{ or } \gap(E) = 0 \text{ and } \rho_{\neg\lambda_0}(E) > 1. \end{cases}
\end{equation}
\end{corollary}

By Corollary~\ref{cor:gap-fraction}, $f_\theta(E)$ does not depend on $p$ for \emph{sufficiently small} $p$. Summing $\Pr(E)\, f_\theta(E)$ over all error configurations and grouping them by weight as in Equation~\eqref{eq:weight-expansion} gives
\begin{equation}\label{eq:ler-expansion}
\LER_\theta(p) = \sum_{E} \Pr(E)\, f_\theta(E) = \sum_{w=0}^{n} F_{\theta,w}\, p^{w} (1-p)^{n-w}, \qquad F_{\theta,w} = \sum_{\wt(E) = w} f_\theta(E)\, b_E.
\end{equation}
Since $(1-p)^{n-w} \to 1$, Lemma~\ref{lem:leading-terms}~(i) and~(ii) apply to Equation~\eqref{eq:ler-expansion} as they do to Equation~\eqref{eq:weight-expansion}.

\begin{definition}[Leading term of the LER]\label{def:ler-leading}
The leading order of the LER of $\theta$ is $o_\theta$, the smallest $w$ with $F_{\theta,w} > 0$, and its leading coefficient is $c_\theta = F_{\theta,o_\theta}$, so that $\LER_\theta(p) \sim c_\theta\, p^{o_\theta}$.
\end{definition}

We now formally define the problem that our predictor addresses:

\begin{problem}[Leading-order LER prediction]\label{prob:leading-ler}
Given a stabilizer code, its set $\mathcal{L}$ of minimum-weight logical operators, its set $\mathcal{N}$ of logical operators of weight $d + 1$ (used only for odd $d$), its per-qubit error biases $b_{i,P}$ and a logical observable $\theta$, compute the leading order $o_\theta$ and the leading coefficient $c_\theta$ of Definition~\ref{def:ler-leading}, without sampling error configurations and without fixing the physical error rate $p$.
\end{problem}

\begin{lemma}[Minimum failing weight]\label{lem:min-weight}
Let $\theta$ be a logical observable. For \emph{sufficiently small} $p$, the ML decoder corrects every error configuration of weight less than $\lceil d/2 \rceil$, so $o_\theta \ge \lceil d/2 \rceil$. If, in addition, some logical operator of weight $d$ flips $\theta$, then $o_\theta = \lceil d/2 \rceil$.
\end{lemma}

Beverland et al.~\cite{beverland_fail_2025} prove this for the minimum-weight decoder, and Appendix~\ref{app:unconcatenated} describes the proof for the ML decoder at leading order and for a single logical observable. By Lemma~\ref{lem:min-weight}, when some logical operator of weight $d$ flips $\theta$, the leading coefficient $c_\theta$ only involves error configurations of weight $w = \lceil d/2 \rceil$. By Corollary~\ref{cor:gap-fraction}, it remains to compute $\gap(E)$, $\Lambda'_E$ and the leading coefficients for each of them. Every error configuration in a competing class with the same syndrome as $E$ is $E \cdot M$ for a non-trivial logical operator $M$, and every error configuration in the class of $E$ is $E \cdot g$ for a stabilizer $g \in S$.

\begin{definition}[Competing logicals]\label{def:competing-logicals}
For an error configuration $E$ of weight $w$, its competing logicals are $\mathcal{C}_E = \{M \in \mathcal{L} \cup \mathcal{N} : \wt(E \cdot M) \le w\}$.
\end{definition}

\begin{lemma}[Competing classes at the leading weight]\label{lem:competitors}
Let $E$ be an error configuration of weight $w$ with $\lambda_0 = \mathsf{cls}(E)$ and syndrome $s$. Then
\begin{enumerate}
    \item[(i)] every $M \in \mathcal{C}_E$ has $\supp(E) \subseteq \supp(M)$, and exactly one of the following holds:
    \begin{enumerate}[nosep]
        \item[(a)] $M \in \mathcal{L}$ and $E \subseteq M$, so that $\wt(E \cdot M) = \lfloor d/2 \rfloor$;
        \item[(b)] $d$ is odd, $M \in \mathcal{L}$, and $M$ differs from $E$ on exactly one qubit of $\supp(E)$, so that $\wt(E \cdot M) = w$;
        \item[(c)] $d$ is odd, $M \in \mathcal{N}$ and $E \subseteq M$, so that $\wt(E \cdot M) = w$;
    \end{enumerate}
    \item[(ii)] the error configurations of weight at most $w$ in competing classes at syndrome $s$ are exactly $E \cdot M$ for $M \in \mathcal{C}_E$, and distinct $M$ give distinct $E \cdot M$;
    \item[(iii)] if $\mathcal{C}_E \neq \emptyset$, then $\omega_{\neg\lambda_0}(s) = \min_{M \in \mathcal{C}_E} \wt(E \cdot M)$ and $\Lambda'_E = \{\lambda_0 \cdot \mathsf{cls}(M) : M \in \mathcal{C}_E,\ \wt(E \cdot M) = \omega_{\neg\lambda_0}(s)\}$, and for each $\lambda' \in \Lambda'_E$, the leading coefficient is
    \begin{equation}\label{eq:competing}
    m_{\lambda'}(s) = \sum_{\substack{M \in \mathcal{C}_E,\ \wt(E \cdot M) = \omega_{\neg\lambda_0}(s) \\ \lambda_0 \cdot \mathsf{cls}(M) = \lambda'}} b_{E \cdot M};
    \end{equation}
    \item[(iv)] if $\mathcal{C}_E = \emptyset$, then $f_\theta(E) = 0$ for every logical observable $\theta$ and \emph{sufficiently small} $p$.
\end{enumerate}
\end{lemma}

\begin{lemma}[True class at the leading weight]\label{lem:true-class}
Let $E$ be an error configuration of weight $w$ with $\lambda_0 = \mathsf{cls}(E)$ and syndrome $s$. Then every error configuration of weight at most $w$ with syndrome $s$ and class $\lambda_0$ is $E \cdot g$ for a stabilizer $g \in S$ of weight at most $2w \le d + 1$. Hence, $\omega_{\lambda_0}(s)$ is the smallest $\wt(E \cdot g)$ over the stabilizers $g$ of weight at most $d + 1$, and
\begin{equation}\label{eq:true-class}
m_{\lambda_0}(s) = \sum_{\substack{g \in S,\ \wt(g) \le d + 1 \\ \wt(E \cdot g) = \omega_{\lambda_0}(s)}} b_{E \cdot g}.
\end{equation}
\end{lemma}

\begin{theorem}[Exact leading term for unconcatenated codes]\label{thm:exact}
Let $\theta$ be a logical observable that some logical operator of weight $d$ flips. Then, for \emph{sufficiently small} $p$, $o_\theta = w$ and
\begin{equation}\label{eq:exact-leading}
c_\theta = \sum_{\substack{\wt(E) = w \\ \mathcal{C}_E \neq \emptyset}} f_\theta(E)\, b_E,
\end{equation}
where $f_\theta(E)$ is given by Corollary~\ref{cor:gap-fraction}, with $\omega_{\neg\lambda_0}(s)$, $\Lambda'_E$ and the leading coefficients of its classes from Lemma~\ref{lem:competitors}, and $\omega_{\lambda_0}(s)$ and $m_{\lambda_0}(s)$ from Lemma~\ref{lem:true-class}.
\end{theorem}

\begin{proof}
By Lemma~\ref{lem:min-weight}, $o_\theta = w$, so by Definition~\ref{def:ler-leading}, $c_\theta = F_{\theta,w}$, which sums $f_\theta(E)\, b_E$ over all error configurations of weight $w$. By Lemma~\ref{lem:competitors}~(iv), the terms with $\mathcal{C}_E = \emptyset$ vanish. For the others, Corollary~\ref{cor:gap-fraction} gives $f_\theta(E)$ from $\gap(E) = \omega_{\neg\lambda_0}(s) - \omega_{\lambda_0}(s)$, $\Lambda'_E$ and the leading coefficients, which Lemma~\ref{lem:competitors}~(iii) and Lemma~\ref{lem:true-class} give.
\end{proof}

Several of these notions follow fail fast~\cite{beverland_fail_2025}: at leading order, the ML decoder decides as its max-class min-weight decoder, and the failure fraction and Equation~\eqref{eq:ler-expansion} are its failure indicator and failure spectrum, resolved per logical observable and with the biases as weights. Case~(a) of Lemma~\ref{lem:competitors} is its Proposition~1 for even $d$, and cases~(a) and~(c) are in its Appendix~A.6 for odd $d$. No step of this subsection uses $k = 1$ or equal biases, so Theorem~\ref{thm:exact} holds for any number of logical qubits and any positive per-qubit biases. Its inputs are $\mathcal{L}$, $\mathcal{N}$ and the stabilizers of weight at most $d + 1$, which we enumerate once per code. Evaluating Equation~\eqref{eq:exact-leading} one error configuration at a time requires listing at least $|\mathcal{L}| \binom{d}{w}$ of them, and Section~\ref{sec:counting} avoids most of this listing.

\subsection{Counting Each Error Configuration Once}\label{sec:counting}
By Lemma~\ref{lem:competitors}~(i), every error configuration in Equation~\eqref{eq:exact-leading} is obtained from some $M \in \mathcal{L} \cup \mathcal{N}$ by keeping $w$ qubits of $\supp(M)$, with at most one of them carrying a different Pauli. Listing these from every such $M$ counts each error configuration once for every logical operator that produces it. For the error configurations inside a minimum-weight logical operator, we avoid both the listing and this overcount.

\begin{definition}[Enclosed error configurations]\label{def:enclosing}
An error configuration $E$ of weight $w$ is \emph{enclosed} if $E \subseteq M$ for some $M \in \mathcal{L}$, following the enclosure terminology of~\cite{beverland_fail_2025}. Its enclosing logicals are $\mathcal{L}_E = \{M \in \mathcal{L} : E \subseteq M\}$.
\end{definition}

\begin{lemma}[Classes at an enclosed configuration]\label{lem:enclosed-classes}
Let $E$ be an \emph{enclosed} error configuration with $\lambda_0 = \mathsf{cls}(E)$ and syndrome $s$. Then
\begin{enumerate}
    \item[(i)] $\omega_{\lambda_0}(s) = w$, $\omega_{\neg\lambda_0}(s) = \lfloor d/2 \rfloor$ and $\Lambda'_E = \{\lambda_0 \cdot \mathsf{cls}(M) : M \in \mathcal{L}_E\}$, so $\gap(E) = \lfloor d/2 \rfloor - w$, which is $0$ for even $d$ and $-1$ for odd $d$;
    \item[(ii)] for odd $d$, all $M \in \mathcal{L}_E$ have the same class, and for \emph{sufficiently small} $p$, $f_\theta(E) = 1$ when they flip $\theta$ and $f_\theta(E) = 0$ otherwise.
\end{enumerate}
\end{lemma}

Different \emph{enclosed} error configurations can have the same enclosing logicals. To compute the failure fraction once for all of them, we group them by the Pauli operator that all their enclosing logicals share.

\begin{definition}[Closure]\label{def:closure}
For Pauli operators $M$ and $M'$, let $M \cap M'$ be the Pauli operator that acts as $P$ on each qubit where both carry the same non-identity Pauli $P$, and as the identity elsewhere. For an \emph{enclosed} error configuration $E$, its closure is $J_E = \bigcap_{M \in \mathcal{L}_E} M$.
\end{definition}

For example, if $d = 4$, $E = X_1X_2$ and $\mathcal{L}_E = \{X_1X_2X_3Z_4,\ X_1X_2X_3Z_5\}$, then $J_E = X_1X_2X_3$, since the two enclosing logicals carry the same Pauli on qubits $1$ to $3$ but not on qubits $4$ and $5$.

Since every $M \in \mathcal{L}_E$ acts as $E$ on $\supp(E)$, and so does every $M \in \mathcal{L}$ with $J_E \subseteq M$, we have,
\begin{equation}\label{eq:closure-logicals}
E \subseteq J_E, \qquad \mathcal{L}_E = \{M \in \mathcal{L} : J_E \subseteq M\}.
\end{equation}

Taking every error configuration $E \subseteq M$ of weight $w$ for every $M \in \mathcal{L}$ counts an \emph{enclosed} error configuration $E$ once for each $M \in \mathcal{L}_E$. We instead count each $E$ once, at its closure $J_E$, using M\"obius inversion~\cite{rota_foundations_1964} over the set $\mathcal{J}$ of intersections $\bigcap_{M \in \mathcal{L}'} M$ over non-empty $\mathcal{L}' \subseteq \mathcal{L}$ that have weight at least $w$, ordered by $\subseteq$. Every closure $J_E = \bigcap_{M \in \mathcal{L}_E} M$ is in $\mathcal{J}$, since $E \subseteq J_E$ gives $\wt(J_E) \ge w$.

\begin{definition}[M\"obius function]\label{def:mobius}
For $K, J \in \mathcal{J}$ with $K \subseteq J$, the M\"obius function is
\begin{equation}\label{eq:mobius}
\mu(K, J) = [K = J] - \sum_{\substack{J' \in \mathcal{J} \\ K \subsetneq J' \subseteq J}} \mu(J', J),
\end{equation}
where $[K = J]$ is $1$ if $K = J$ and $0$ otherwise.
\end{definition}

\begin{lemma}[Counting enclosed configurations once]\label{lem:mobius-count}
Let $\varphi$ assign a real number $\varphi(J, E)$ to every $J \in \mathcal{J}$ and every error configuration $E \subseteq J$ of weight $w$. Then
\begin{equation}\label{eq:mobius-count}
\sum_{E\in \{E:\ E\text{ is \emph{enclosed}}\}} \varphi(J_E, E) = \sum_{J \in \mathcal{J}} \; \sum_{\substack{J' \in \mathcal{J} \\ J' \subseteq J}} \mu(J', J) \left( \sum_{\substack{E \subseteq J' \\ \wt(E) = w}} \varphi(J, E) \right).
\end{equation}
\end{lemma}

To count the failures of the \emph{enclosed} error configurations, we apply Lemma~\ref{lem:mobius-count} with $\varphi(J, E) = f^J_\theta(E)\, b_E$, where $f^J_\theta(E)$ is given by Corollary~\ref{cor:gap-fraction}, with $\Lambda'_E$ and its leading coefficients computed by Lemmas~\ref{lem:enclosed-classes}~(i) and~\ref{lem:competitors}~(iii) from $\mathcal{L}_J = \{M \in \mathcal{L} : J \subseteq M\}$ in place of $\mathcal{L}_E$, and $m_{\lambda_0}(s)$ from Lemma~\ref{lem:true-class}. By Equation~\eqref{eq:closure-logicals}, $\mathcal{L}_{J_E} = \mathcal{L}_E$, so $\varphi(J_E, E) = f_\theta(E)\, b_E$. For odd $d$, Lemma~\ref{lem:enclosed-classes}~(ii) gives $f^J_\theta(E)$ directly, while for even $d$, $f^J_\theta(E)$ depends on $m_{\lambda_0}(s)$, which can differ between \emph{enclosed} error configurations with the same closure. Section~\ref{sec:dp} evaluates the inner sums of Equation~\eqref{eq:mobius-count} without listing their error configurations.

The remaining error configurations of Equation~\eqref{eq:exact-leading} are not \emph{enclosed}, exist only for odd $d$, and we list them directly. For each $M \in \mathcal{L}$, we take every set of $w$ qubits of $\supp(M)$, keep the Pauli of $M$ on all but one of them, and give the remaining qubit each of the two other non-identity Paulis. For each $N \in \mathcal{N}$, we take every error configuration $E \subseteq N$ of weight $w$. We keep each distinct error configuration once, discard the \emph{enclosed} ones, and sum $f_\theta(E)\, b_E$ over the rest, with $f_\theta(E)$ as in Theorem~\ref{thm:exact}.

\section{Exact Leading-Order LER for Concatenated Codes}\label{sec:concatenated}
Our main result for concatenated codes, Theorem~\ref{thm:recursion}, extends Theorem~\ref{thm:exact} from a single code block to every level. The soft information that a block passes to its parent has a distribution whose leading terms are determined by those of its child blocks. Starting from the physical qubits, this gives the exact leading term of the LER at every level. In a concatenated code, the tree decoder of Section~\ref{sec:tree-decoder} passes the class likelihoods of each child block to its parent block, which uses them to decode its own syndrome. Thus, the LER of a parent block depends on how often each child block produces each value of these likelihoods. Section~\ref{sec:soft-distribution} defines this as the soft-information distribution of a block, and Section~\ref{sec:recursion} computes the distribution of a parent block from those of its children and gives the leading order of the LER in closed form. For simplicity, every block in this section has a single logical qubit, i.e., $k = 1$, so that by Theorem~\ref{thm:tree-ml} the tree decoder is the ML decoder of the whole concatenated code. Theorem~\ref{thm:recursion} also holds for $k > 1$ under interleaved concatenation, with the soft information of each logical qubit in place of that of the block. No block then contains two logical qubits of the same block below it, so the inputs of every block remain independent, but Theorem~\ref{thm:tree-ml} no longer applies. The proofs of this section are in Appendix~\ref{app:concatenated}.

\subsection{Soft-Information Distribution}\label{sec:soft-distribution}
In the tree of Section~\ref{sec:tree-decoder}, we call a physical qubit a block at level $0$, and every block at level $\ell \ge 1$ encodes its logical qubit using $n$ qubits. Together with the blocks below it, a block at level $\ell$ is a stabilizer code with one logical qubit on the $N = n^{\ell}$ physical qubits below it. Hence, Section~\ref{sec:leading-order} and Corollary~\ref{cor:gap-fraction} apply to it with $N$ in place of $n$. For $P \in \{X, Y, Z\}$, let $\bar P$ be a logical operator of the block that acts as $P$ on its logical qubit. The competing classes of an error configuration $E$ with $\lambda_0 = \mathsf{cls}(E)$ are then the three classes $\lambda_0 \cdot \mathsf{cls}(\bar P)$. We now describe, at leading order, the soft information that a block passes to its parent in the tree decoder of Section~\ref{sec:tree-decoder}.

\begin{definition}[Soft information]\label{def:soft-information}
Let $E$ be an error configuration of a block with $\lambda_0 = \mathsf{cls}(E)$. For $P \in \{X, Y, Z\}$, let $\gap_P(E)$ and $\rho_P(E)$ be the signed gap and the coefficient ratio of Definition~\ref{def:signed-gap} at the competing class $\lambda_0 \cdot \mathsf{cls}(\bar P)$. The soft information of $E$ is the triple $\sigma(E) = \left( \gap_P(E), \rho_P(E) \right)_{P \in \{X, Y, Z\}}$.
\end{definition}

By Equation~\eqref{eq:class-ratio}, the soft information gives the leading term of each likelihood ratio that the block passes up, i.e., for every $P \in \{X, Y, Z\}$,
\begin{equation}\label{eq:soft-ratio}
\frac{W_{\lambda_0 \cdot \mathsf{cls}(\bar P)}(s)}{W_{\lambda_0}(s)} \sim \rho_P(E)\, p^{\gap_P(E)},
\end{equation}
where $s$ is the syndrome of $E$. The soft information also decides whether the block itself fails. For \emph{sufficiently small} $p$, the failure fraction $f_\theta(E)$ of a logical observable $\theta$ depends on $E$ only through $\sigma(E)$. Indeed, Corollary~\ref{cor:gap-fraction} reads the classes it compares off the signed gaps and coefficient ratios in $\sigma(E)$, and returning the class $\lambda_0 \cdot \mathsf{cls}(\bar P)$ fails for $\theta$ exactly when $\bar P$ flips $\theta$. For a value $\tau$ of the soft information, we write $f_\theta(\tau)$ for the value of $f_\theta(E)$ when $\sigma(E) = \tau$. Hence, the LER depends only on how likely each value of the soft information is.

\begin{definition}[Soft-information distribution]\label{def:soft-distribution}
For a value $\tau$ of the soft information of a block, let $\Pr(\sigma(E) = \tau)$ be the probability that the error configuration $E$ of the block has soft information $\tau$. The soft-information distribution $\mathcal{D}$ of the block gives, for each such $\tau$, the leading order and the leading coefficient of this probability, $\mathcal{D}(\tau) = \left( o(\tau), q(\tau) \right)$, where
\begin{equation}\label{eq:soft-distribution}
o(\tau) = \min_{\sigma(E) = \tau} \wt(E), \qquad q(\tau) = \sum_{\substack{\sigma(E) = \tau \\ \wt(E) = o(\tau)}} b_E.
\end{equation}
\end{definition}

By Equation~\eqref{eq:configuration-prob}, $\Pr(E) \sim b_E\, p^{\wt(E)}$. Since $\Pr(\sigma(E) = \tau)$ is the sum of $\Pr(E)$ over the error configurations $E$ with $\sigma(E) = \tau$, Lemma~\ref{lem:leading-terms}~(i) and~(ii) give,
\begin{equation}\label{eq:soft-leading}
\Pr(\sigma(E) = \tau) \sim q(\tau)\, p^{o(\tau)}.
\end{equation}

\sloppy Grouping the error configurations in Equation~\eqref{eq:ler-expansion} by their soft information gives $\LER_\theta(p) = \sum_{\tau \,:\, f_\theta(\tau) > 0} f_\theta(\tau) \Pr(\sigma(E) = \tau)$. By Equation~\eqref{eq:soft-leading} and Lemma~\ref{lem:leading-terms}, the leading order and the leading coefficient of the LER of a logical observable $\theta$ are then, for \emph{sufficiently small} $p$,
\begin{equation}\label{eq:soft-ler}
o_\theta = \min \{ o(\tau) : f_\theta(\tau) > 0 \}, \qquad c_\theta = \sum_{\tau \,:\, o(\tau) = o_\theta} f_\theta(\tau)\, q(\tau).
\end{equation}

A physical qubit $i$ is a block with $N = 1$ whose only stabilizer is the identity. Its syndrome is therefore always trivial, and $\bar P$ is $P$ itself. Each error configuration $Q \in \{I, X, Y, Z\}$ defines a logical class by itself, so $W_{\mathsf{cls}(Q)}(s) = \Pr(Q)$. For $E = Q$, the class $\lambda_0 \cdot \mathsf{cls}(\bar P)$ is $\mathsf{cls}(QP)$. The likelihood ratio in Equation~\eqref{eq:soft-ratio} is then $\Pr(QP)/\Pr(Q)$, with $\Pr(I) = 1 - p$ and $\Pr(Q') = b_{i,Q'}\, p$ for $Q' \neq I$. Its leading term gives,
\begin{equation}\label{eq:qubit-soft}
\left( \gap_P(Q), \rho_P(Q) \right) = \begin{cases} \left( 1, b_{i,P} \right) & \text{if } Q = I, \\ \left( -1, 1/b_{i,Q} \right) & \text{if } Q = P, \\ \left( 0, b_{i,QP}/b_{i,Q} \right) & \text{otherwise}. \end{cases}
\end{equation}
The four values $\sigma(I)$, $\sigma(X)$, $\sigma(Y)$ and $\sigma(Z)$ are distinct, so by Definition~\ref{def:soft-distribution}, the soft-information distribution of a physical qubit is $\mathcal{D}(\sigma(I)) = (0, 1)$ and $\mathcal{D}(\sigma(Q)) = (1, b_{i,Q})$ for $Q \in \{X, Y, Z\}$. Starting from these, Section~\ref{sec:recursion} computes the soft-information distribution of every block from those of its children.

\subsection{Exact Recursion}\label{sec:recursion}
We now show that the soft information of a parent block is a function of that of its child blocks (Lemma~\ref{lem:parent-soft}), and that, since the error configurations of the child blocks are independent, the soft-information distribution of the parent block follows from those of its child blocks (Theorem~\ref{thm:recursion}). The parent block encodes its logical qubit using $n$ qubits with a stabilizer code, and qubit $j$ comes from the logical qubit of child block $j$. Let $S$ be the stabilizer group of this code. For $P \in \{I, X, Y, Z\}$, let $\mathcal{M}_P$ be the set of logical operators $M$ of this code that act as $P$ on its logical qubit, so that $\mathcal{M}_I = S$.

Since the child blocks act on disjoint physical qubits, every error configuration of the parent block is a product $E = E^{(1)} \cdots E^{(n)}$, where $E^{(j)}$ is an error configuration of child block $j$. We write $\lambda_j = \mathsf{cls}(E^{(j)})$, $s^{(j)}$ for the syndrome of $E^{(j)}$, and $W^{(j)}$ for the class likelihoods of child block $j$. Let $\lambda_0 = \mathsf{cls}(E)$ and let $s$ be the syndrome of $E$. The errors on the qubits of the parent block with syndrome $s$ and class $\lambda_0 \cdot \mathsf{cls}(\bar P)$ are exactly the products of the $M \in \mathcal{M}_P$ with the error that applies, on each qubit $j$, the Pauli of class $\lambda_j$, and multiplying by $M$ changes the class of child block $j$ from $\lambda_j$ to $\lambda_j \cdot \mathsf{cls}(\bar{M_j})$. Equation~\eqref{eq:tree-likelihood} thus gives, for $P \in \{I, X, Y, Z\}$,
\begin{equation}\label{eq:parent-likelihood}
W_{\lambda_0 \cdot \mathsf{cls}(\bar P)}(s) = \sum_{M \in \mathcal{M}_P} \; \prod_{j=1}^{n} W^{(j)}_{\lambda_j \cdot \mathsf{cls}(\bar{M_j})}\left( s^{(j)} \right),
\end{equation}
where $\bar I$ is the identity, and $\bar{M_j}$ is the logical operator $\bar P$ of child block $j$ for $P = M_j$. By Theorem~\ref{thm:tree-ml}, these are the class likelihoods of the parent block as a code on its $N$ physical qubits.

\begin{lemma}[Soft information of a parent block]\label{lem:parent-soft}
Let $E = E^{(1)} \cdots E^{(n)}$ be an error configuration of a parent block. Then, for every $P \in \{X, Y, Z\}$,
\begin{equation}\label{eq:parent-soft}
\rho_P(E)\, p^{\gap_P(E)} \sim \frac{\sum_{M \in \mathcal{M}_P} \prod_{j \in \supp(M)} \rho_{M_j}(E^{(j)})\, p^{\gap_{M_j}(E^{(j)})}}{\sum_{g \in S} \prod_{j \in \supp(g)} \rho_{g_j}(E^{(j)})\, p^{\gap_{g_j}(E^{(j)})}}.
\end{equation}
Hence, $\sigma(E) = \Phi\left( \sigma(E^{(1)}), \ldots, \sigma(E^{(n)}) \right)$, where the map $\Phi$ reads $\gap_P(E)$ and $\rho_P(E)$ off the leading term of the right side of Equation~\eqref{eq:parent-soft} for each $P$, and depends only on the code of the parent block.
\end{lemma}

Expanding the products in Equation~\eqref{eq:parent-soft}, each $M$ contributes a single power of $p$, whose exponent is its accumulated gap $\sum_{j \in \supp(M)} \gap_{M_j}(E^{(j)})$, with the product of the $\rho_{M_j}(E^{(j)})$ as its coefficient. By Lemma~\ref{lem:leading-terms}~(ii), $\gap_P(E)$ is therefore the smallest accumulated gap over $\mathcal{M}_P$ minus the smallest over $S$, and $\rho_P(E)$ is the ratio of the summed coefficients at these two smallest values.

\begin{theorem}[Exact recursion]\label{thm:recursion}
For a parent block, let $\mathcal{D}^{(j)}(\tau_j) = \left( o^{(j)}(\tau_j), q^{(j)}(\tau_j) \right)$ be the soft-information distribution of child block $j$. Then, for every value $\tau$ of the soft information of the parent block,
\begin{equation}\label{eq:recursion}
\begin{aligned}
\Pr(\sigma(E) = \tau) &= \sum_{\Phi(\tau_1, \ldots, \tau_n) = \tau} \; \prod_{j=1}^{n} \Pr\left( \sigma(E^{(j)}) = \tau_j \right), \\
q(\tau)\, p^{o(\tau)} &\sim \sum_{\Phi(\tau_1, \ldots, \tau_n) = \tau} \; \prod_{j=1}^{n} q^{(j)}(\tau_j)\, p^{o^{(j)}(\tau_j)},
\end{aligned}
\end{equation}
where each $\tau_j$ runs over the values of the soft information of child block $j$. In particular, $\mathcal{D}$ is determined by $\mathcal{D}^{(1)}, \ldots, \mathcal{D}^{(n)}$.
\end{theorem}

The proof uses that the child blocks are independent, and groups the error configurations of each child block by their soft information. Theorem~\ref{thm:recursion} needs only the distributions $\mathcal{D}^{(j)}$ of the child blocks, not their error configurations. Starting from the physical qubits, whose $\mathcal{D}$ is given at the end of Section~\ref{sec:soft-distribution}, it gives $\mathcal{D}$ for every block, one level at a time. By Equation~\eqref{eq:soft-ler}, this gives the exact leading term of the LER at every level, for \emph{sufficiently small} $p$.

For $P \in \{X, Y, Z\}$, let $d(P)$ be the smallest weight of a logical operator of the block that acts as $P$ on its logical qubit. Then $\gap_P(I) = d(P)$, and applying Lemma~\ref{lem:parent-soft} to $E = I$ gives,
\begin{equation}\label{eq:class-distance}
d(P) = \min_{M \in \mathcal{M}_P} \; \sum_{j \in \supp(M)} d^{(j)}(M_j),
\end{equation}
where $d^{(j)}$ are the class distances of child block $j$. As for unconcatenated codes in Lemma~\ref{lem:min-weight}, the leading order of the LER of $\theta$ is then $\lceil D_\theta/2 \rceil$, where $D_\theta = \min \{ d(P) : \bar P \text{ flips } \theta \}$, and no decoder has a higher order~\cite{beverland_fail_2025}. With the same code at every level and $d(X) = d(Y) = d(Z) = d$ at level $1$, Equation~\eqref{eq:class-distance} gives $D_\theta = d^{\ell}$ at level $\ell$~\cite{knill_concatenated_1996}. Evaluating Equation~\eqref{eq:recursion} directly is exponential in $n$. The sums in Equation~\eqref{eq:parent-soft} run over the $2^{n-1}$ elements of each $\mathcal{M}_P$, and the sums in Equation~\eqref{eq:recursion} run over all tuples of values of the child blocks. At level 1, Theorem~\ref{thm:exact} avoids this by counting only the error configurations of weight $\lceil d/2 \rceil$. Section~\ref{sec:approximation} approximates Equation~\eqref{eq:recursion} above level 1.

\section{Staying Sober: Budgeted DP to Estimate LERs}\label{sec:approximation}
This section describes our predictor, which approximates Theorem~\ref{thm:recursion} above level 1. None of the relaxations changes the leading term at level 1, so if no budget of Section~\ref{sec:cost} is exhausted, the predictor is still exact at level 1 by Theorem~\ref{thm:exact}. For fixed budgets, the predictor takes time polynomial in $d$ and in the number of minimum-weight logical operators at each level, so its runtime is linear in the number of levels. Section~\ref{sec:relaxed} identifies the relaxations we make to the original requirements in Theorem~\ref{thm:exact} and Theorem~\ref{thm:recursion}, Section~\ref{sec:dp} then describes our algorithm, Section~\ref{sec:cost} identifies the budgeted approximations, and Section~\ref{sec:evaluation} discusses the accuracy of the predictor.

\subsection{Relaxed Requirements}\label{sec:relaxed}
Before describing the relaxed requirements, we first define a slightly different distribution that the child blocks pass up to their parents, which in practice turns out to be easier to compute.

\begin{definition}[Distribution restricted to a Pauli]\label{def:pauli-distribution}
Let $P \in \{X, Y, Z\}$. For an error configuration $E$ of a block with $\lambda_0 = \mathsf{cls}(E)$, the soft information of $E$ restricted to $P$ is $\sigma_P(E) = \left( \gap_P(E), \rho_P(E), [\lambda_0 \cdot \mathsf{cls}(\bar P) \in \Lambda'_E] \right)$, where $[\cdot]$ is $1$ if the condition holds and $0$ otherwise. The distribution $\mathcal{D}_P$ of the block restricted to $P$ is defined as $\mathcal{D}$ in Definition~\ref{def:soft-distribution}, with $\sigma_P(E)$ in place of $\sigma(E)$, i.e., $\mathcal{D}_P(\tau) = \left( o_P(\tau), q_P(\tau) \right)$ gives the leading order and the leading coefficient of $\Pr(\sigma_P(E) = \tau)$.
\end{definition}

As in Equation~\eqref{eq:soft-leading}, $\Pr(\sigma_P(E) = \tau) \sim q_P(\tau)\, p^{o_P(\tau)}$. We also name the child blocks on which an error configuration of a parent block acts.

\begin{definition}[Touched child blocks]\label{def:touched}
An error configuration $E = E^{(1)} \cdots E^{(n)}$ of a parent block touches child block $j$ if $E^{(j)} \neq I$. The set of child blocks that $E$ touches is $Q_E = \{ j : E^{(j)} \neq I \}$.
\end{definition}

Theorem~\ref{thm:recursion} needs the soft-information distribution $\mathcal{D}^{(j)}$ of every child block $j$, and sums over every error configuration of the parent block. The predictor relaxes these two requirements.
\begin{enumerate}
    \item[(R1)] \emph{Marginals.} Each child block passes up $\mathcal{D}_X$, $\mathcal{D}_Y$ and $\mathcal{D}_Z$ separately, instead of $\mathcal{D}$. These lose the correlations between the Paulis, e.g., how likely an error configuration with $\gap_X(E) = 0$ is to also have $\gap_Y(E) = 0$, except for the one bit that the last entry of $\sigma_P(E)$ keeps, since keeping them all would require listing all $2^d$ error configurations inside a single logical operator of weight $d$. Section~\ref{sec:dp} approximates the three quantities that need these correlations: the true class, i.e., the denominator of Equation~\eqref{eq:parent-soft}, the soft information at the Paulis of the non-lightest competing classes, and the ambiguous classes in Corollary~\ref{cor:gap-fraction}.
    \item[(R2)] \emph{Truncation.} Instead of summing over every error configuration of the parent block, the predictor keeps only those with $Q_E \subseteq \supp(M)$ for some $M \in \mathcal{L}$, and, at level 1 and for odd $d$, those of cases~(b) and~(c) of Lemma~\ref{lem:competitors}. A parent block may need error configurations of its child blocks of weight up to $d$ that no element of $\mathcal{L}$ contains. However, enumerating every error configuration of weight $d$ in a block of $n$ qubits would require iterating over $\binom{n}{d} 3^d$ error configurations.
\end{enumerate}

\subsection{Counting by Accumulated Gap}\label{sec:dp}
We now describe how the predictor computes the distributions $\mathcal{D}_X$, $\mathcal{D}_Y$ and $\mathcal{D}_Z$, and the LER, of a parent block from the distributions of its child blocks, using the notation of Section~\ref{sec:recursion}. By (R2), it counts each error configuration $E$ with $Q_E \subseteq \supp(M)$ for some $M \in \mathcal{L}$ once, through its closure and Lemma~\ref{lem:mobius-count}, as in Section~\ref{sec:counting}. Therefore, for each closure $J \in \mathcal{J}$ and each $J' \subseteq J$ in $\mathcal{J}$, we need to evaluate the inner sum of Equation~\eqref{eq:mobius-count}, which runs over the error configurations $E$ with $Q_E \subseteq \supp(J')$.

For an error configuration $E$ that has $Q_E \subseteq \supp(J)$, Equation~\eqref{eq:parent-soft} becomes simpler. Every $M \in \mathcal{L}_J$ queries each child block $j \in \supp(J)$ at $M_j = J_j$. Since $Q_E \subseteq \supp(J)$, the remaining child blocks $j \in \supp(M) \setminus \supp(J)$ are untouched, so their soft information at $M_j$ is $\gap_{M_j}(I) = d^{(j)}(M_j)$ with coefficient ratio $\rho_{M_j}(I)$ by Section~\ref{sec:recursion}. The term of $M$ in the numerator of Equation~\eqref{eq:parent-soft} therefore factors as,
\begin{equation}\label{eq:closure-factor}
\left( \prod_{j \in \supp(J)} \rho_{J_j}(E^{(j)})\, p^{\gap_{J_j}(E^{(j)})} \right) \left( \prod_{j \in \supp(M) \setminus \supp(J)} \rho_{M_j}(I)\, p^{d^{(j)}(M_j)} \right),
\end{equation}
where the first factor depends only on $E$, and the second depends only on $M$. Summing the second factor, a product over the child blocks in $\supp(M) \setminus \supp(J)$, over the $M \in \mathcal{L}_J$ of one class gives a leading term, which we call the leading factor of that class. Let $r_J\, p^{e_J}$ be the largest of these leading factors for \emph{sufficiently small} $p$, and let $\mathcal{P}_J$ be the set of Paulis $P$ for which the leading factor of the class $\mathsf{cls}(\bar P)$ is $r_J\, p^{e_J}$. The predictor compares the true class only with these most likely competing classes $\lambda_0 \cdot \mathsf{cls}(\bar P)$, $P \in \mathcal{P}_J$. By (R1), it only estimates the soft information at the other Paulis, from the logical operators that share child blocks with $E$ and from the values of $E$ at those Paulis. Before describing the algorithm, we evaluate the true class, i.e., the denominator of Equation~\eqref{eq:parent-soft}, from the marginals of (R1), using the following lemma.

\begin{lemma}[Zero difference is a containment]\label{lem:zero-difference}
Let $M \in \mathcal{L}$, let $F \subseteq M$ be an error configuration, and let $g \in S$. Then $\wt(F \cdot g) \ge \wt(F)$, with equality exactly when $\wt(M \cdot g) = d$ and $\supp(g) \cap \supp(M) \subseteq \supp(F)$.
\end{lemma}

We prove Lemma~\ref{lem:zero-difference} in Appendix~\ref{app:predictor}. In the denominator, a stabilizer may act upon a touched child block $j$ at a Pauli other than $J_j$, where (R1) does not provide the soft information of $E^{(j)}$. The predictor therefore assumes that every touched child block $j$ carries a lightest logical operator of class $\mathsf{cls}(\bar{J_j})$, and considers every child block as a single qubit, on which a Pauli $P$ occurs with weight $\rho_P(I)$. The parent block then becomes an unconcatenated code on its $n$ child blocks, and $E$ becomes the error configuration $\tilde E$ that acts as $J$ on $Q_E$ and as the identity elsewhere. Since $\tilde E \subseteq J \subseteq M$ for every $M \in \mathcal{L}_J$, Lemma~\ref{lem:zero-difference} shows that no stabilizer lowers the weight of $\tilde E$. Thus, as in Lemma~\ref{lem:true-class}, the leading term of the denominator of Equation~\eqref{eq:parent-soft} becomes,
\begin{equation}\label{eq:true-coefficient}
\eta(E) = \sum_{\substack{g \in S \\ \wt(\tilde E \cdot g) = \wt(\tilde E)}} \frac{b'_{\tilde E \cdot g}}{b'_{\tilde E}}, \qquad b'_F = \prod_{j \in \supp(F)} \rho_{F_j}(I),
\end{equation}
where $b'_F$ is the bias product of Equation~\eqref{eq:configuration-prob} with $\rho_P(I)$ in place of the biases. At level 1, $\tilde E = E$ and $\rho_P(I) = b_{j,P}$ by Equation~\eqref{eq:qubit-soft}, so $\eta(E) = m_{\lambda_0}(s) / b_E$ is exact. If $E$ touches all $d$ child blocks of some $M \in \mathcal{L}$, then $\tilde E = M$, and by Lemma~\ref{lem:zero-difference}, the other terms of Equation~\eqref{eq:true-coefficient} are other elements of $\mathcal{L}$, whose error configurations the predictor counts separately, through their own closures. It then sets $\eta(E) = 1$. If, in addition, the soft information $\sigma_{J_j}(E^{(j)})$ of every touched child block $j$ equals that of a minimum-weight logical operator of class $\mathsf{cls}(\bar{J_j})$, then $E$ has the soft information of a logical operator of the parent block, and the predictor uses the exact coefficient ratio $1/\rho_P(I)$ of that logical operator, where $\mathsf{cls}(\bar P) = \mathsf{cls}(M)$.

The coefficient $\eta(E)$ depends on which child blocks $E$ touches, and not only on how many. Fix any $M \in \mathcal{L}_J$. By Lemma~\ref{lem:zero-difference} with $F = \tilde E$, a stabilizer $g$ appears in Equation~\eqref{eq:true-coefficient} exactly when $\wt(M \cdot g) = d$ and its overlap $\supp(g) \cap \supp(M)$ lies inside $Q_E$. The term of such a $g$ depends on $\tilde E$ only through this overlap, where $\tilde E$ acts as $J$. Hence, $\eta(E)$ depends on $E$ only through the overlaps that lie inside $Q_E$.

Algorithm~\ref{alg:dp} evaluates the inner sum of Equation~\eqref{eq:mobius-count}. At level 1, each qubit takes only the value $\sigma_{J_j}(I)$ or $\sigma_{J_j}(J_j)$, so that $E \subseteq J'$, as in Section~\ref{sec:counting}. We call child blocks with the same $J_j$ and the same $\mathcal{D}^{(j)}_{J_j}$ that lie in the same overlaps \emph{identical}, since they contribute identically to the first factor of Equation~\eqref{eq:closure-factor} and to $\eta(E)$. We write the three entries of a value $\tau$ of $\sigma_P(E)$ as $\tau.\texttt{gap}$, $\tau.\texttt{ratio}$ and $\tau.\texttt{isLightest}$. For a set of \emph{identical} child blocks, \texttt{Splits} lists every way to choose how many of them have each soft information $\tau \neq \sigma_P(I)$ of $\mathcal{D}_P$, leaving the rest untouched, and \texttt{Ways} counts the ways to choose which child blocks those are.

\begin{algorithm}[t]
    \caption{DP over the child blocks of $J'$ for the closure $J$}
    \label{alg:dp}
    \begin{algorithmic}[1]
        \Procedure{CountByGap}{$J, J'$}
            \State{prob $\gets \{\}$\Comment{maps each state to the leading term of its probability, and is $0$ at a missing state}}
            \State{start.gap $\gets \sum_{j \in \supp(J)} d^{(j)}(J_j)$, \quad start.ratio $\gets \prod_{j \in \supp(J)} \rho_{J_j}(I)$}
            \State{start.numTouched $\gets 0$, \quad start.numNotLightest $\gets 0$, \quad start.fullyTouched $\gets \emptyset$}
            \State{prob[start] $\gets 1$\Comment{initialization: $E = I$ touches no child block}}
            \ForAll{blocks $\in$ \Call{IdenticalBlocks}{$J'$}\Comment{\emph{identical} child blocks $j \in \supp(J')$}}
                \State{$P \gets$ blocks.pauli, \quad $\mathcal{D}_P \gets$ blocks.distribution}
                \State{noErrorGap $\gets \sigma_P(I)$.gap, \quad noErrorRatio $\gets \sigma_P(I)$.ratio}
                \State{newProb $\gets \{\}$}
                \ForAll{state $\in$ prob}
                    \ForAll{counts $\in$ \Call{Splits}{blocks}}
                        \Statex{\Comment{counts[$\tau$] of these child blocks take the value $\tau \neq \sigma_P(I)$ of $\mathcal{D}_P$}}
                        \State{next.gap $\gets$ state.gap $+ \sum_\tau$ counts[$\tau$] $\times$ ($\tau$.gap $-$ noErrorGap)}
                        \State{next.ratio $\gets$ state.ratio $\times \prod_\tau$ ($\tau$.ratio $/$ noErrorRatio)$^{\text{counts}[\tau]}$}
                        \State{next.numTouched $\gets$ state.numTouched $+ \sum_\tau$ counts[$\tau$]}
                        \State{next.numNotLightest $\gets$ state.numNotLightest $+ \sum_{\tau :\, \tau.\text{isLightest} = 0}$ counts[$\tau$]}
                        \State{next.fullyTouched $\gets$ state.fullyTouched $\cup$ (blocks if $\sum_\tau$ counts[$\tau$] $=$ $|$blocks$|$, else $\emptyset$)}
                        \State{term $\gets$ prob[state] $\times$ \Call{Ways}{blocks, counts} $\times \prod_\tau \left( q_P(\tau)\, p^{o_P(\tau)} \right)^{\text{counts}[\tau]}$}
                        \State{newProb[next] $\gets$ newProb[next] $+$ term}
                        \Statex{\Comment{update: \texttt{Ways} counts the choices of which child blocks take each value}}
                    \EndFor{}
                \EndFor{}
                \State{prob $\gets$ newProb}
            \EndFor{}
            \State{\Return{prob}}
        \EndProcedure{}
    \end{algorithmic}
\end{algorithm}

\sloppy A state describes the error configurations $E$ with the same first factor of Equation~\eqref{eq:closure-factor}, namely $\texttt{state.ratio} \times p^{\texttt{state.gap}}$, the same number \texttt{state.numTouched} of touched child blocks, the same number \texttt{state.numNotLightest} of touched child blocks whose picked value $\tau$ has $\tau.\texttt{isLightest} = 0$, and the same set \texttt{state.fullyTouched} of child blocks in sets of \emph{identical} child blocks that they touch completely. Since every overlap is a union of such sets, an overlap lies inside $Q_E$ exactly when it lies inside \texttt{state.fullyTouched}. Thus, all these error configurations have the same $\eta(E)$, which we write as $\eta(\texttt{state})$. Since $\Pr(E) = \prod_j \Pr(E^{(j)})$, as in Theorem~\ref{thm:recursion}, $\texttt{prob[state]}$ is asymptotic to the total probability $\sum_E \Pr(E)$ of these error configurations.

For each state, Algorithm~\ref{alg:level} estimates the soft information of its error configurations at every $P \in \mathcal{P}_J$ from Equations~\eqref{eq:closure-factor} and~\eqref{eq:true-coefficient} as,
\begin{equation}\label{eq:state-soft}
\gap_P(E) = \texttt{state.gap} + e_J, \qquad \rho_P(E) = \frac{\texttt{state.ratio} \times r_J}{\eta(\texttt{state})},
\end{equation}
with the replacements above for the error configurations that touch all $d$ child blocks of some $M \in \mathcal{L}$. By Corollary~\ref{cor:gap-fraction}, with the classes $\lambda_0 \cdot \mathsf{cls}(\bar P)$ for $P \in \mathcal{P}_J$ in place of $\Lambda_E$, the decoder of the parent block then returns each of these classes with probability,
\begin{equation}\label{eq:return-share}
\begin{cases} 1/|\mathcal{P}_J| & \text{if } \gap_P(E) < 0, \text{ or } \gap_P(E) = 0 \text{ and } \rho_P(E) > 1, \\ 1/(|\mathcal{P}_J| + \texttt{state.numNotLightest} + 1) & \text{if } \gap_P(E) = 0 \text{ and } \rho_P(E) = 1, \\ 0 & \text{otherwise}. \end{cases}
\end{equation}
In the middle case, each touched child block whose picked value has last entry $0$ may add an ambiguous class that cannot be identified due to (R1), and the predictor assumes an ambiguous class exists for each such child block.

\begin{algorithm}[t]
    \caption{One level of the predictor}
    \label{alg:level}
    \begin{algorithmic}[1]
        \Procedure{ParentBlock}{$\mathcal{J}$}
            \State{dist $\gets \{\}$, \quad ler $\gets \{\}$}
            \Statex{\Comment{dist[$P$] maps each value of $\sigma_P$ to the leading term of its probability, ler[$\theta$] is the leading term of the LER of $\theta$, and both resolve their signed sums from the lowest order up}}
            \ForAll{$J, J' \in \mathcal{J}$ with $J' \subseteq J$}
                \State{prob $\gets$ \Call{CountByGap}{$J, J'$}}
                \ForAll{state $\in$ prob with state.numTouched $\ge \lceil d/2 \rceil$}
                    \Statex{\Comment{states with fewer touched child blocks are left to \texttt{AddListed}}}
                    \State{weight $\gets \mu(J', J) \times$ prob[state]}
                    \State{$\eta$(state) $\gets$ Equation~\eqref{eq:true-coefficient}, over the stabilizers whose overlap lies in state.fullyTouched}
                    \Statex{\Comment{cached for each $J$ and state.fullyTouched}}
                    \State{gap, ratio $\gets$ Equation~\eqref{eq:state-soft}, \quad share $\gets$ Equation~\eqref{eq:return-share}}
                    \ForAll{$P \in \mathcal{P}_J$}
                        \If{gap $\le 0$ or state.numNotLightest $= 0$}
                            \State{dist[$P$][(gap, ratio, $1$)] $\gets$ dist[$P$][(gap, ratio, $1$)] $+$ weight}
                            \Statex{\Comment{otherwise, the closure that queries the non-lightest child blocks at their lightest classes adds it}}
                        \EndIf
                    \EndFor
                    \ForAll{logical observables $\theta$}
                        \State{ler[$\theta$] $\gets$ ler[$\theta$] $+$ share $\times |\{P \in \mathcal{P}_J : \bar P \text{ flips } \theta\}| \times$ weight}
                    \EndFor
                \EndFor
            \EndFor
            \State{\Call{AddListed}{dist, ler}}
            \Statex{\Comment{the error configurations that touch fewer than $\lceil d/2 \rceil$ child blocks, those of cases~(b) and~(c) at level 1, and the estimates at the other Paulis}}
            \State{\Return{dist, ler}}
        \EndProcedure{}
    \end{algorithmic}
\end{algorithm}

\subsection{Budgets and Runtime Cost}\label{sec:cost}
\begin{table}[t]
    \centering
    \small
    \caption{Budgets of our predictor, a description of what they budget for, and their default values.}
    \label{tab:budgets}
    \begin{tabular}{l p{0.75\textwidth} r}
        \toprule
        Budget & Description & Default \\
        \midrule
        $B_{\text{closure}}$ & elements of $\mathcal{J}$, or error configurations sampled & $2 \times 10^{5}$ \\
        $B_{\text{gap}}$ & gaps stored in each $\mathcal{D}_P$, and $2B_{\text{gap}} + 1$ accumulated gaps in Algorithm~\ref{alg:dp}, smallest magnitude first & $16$ \\
        $B_{\text{likelihood}}$ & coefficient ratios stored per gap in each $\mathcal{D}_P$ & $16$ \\
        $B_{\text{cellratio}}$ & values of \texttt{state.ratio} stored per value of the other fields of a state & $4096$ \\
        $B_{\text{work}}$ & steps of the enumeration of stabilizers & $3 \times 10^{6}$ \\
        $B_{\text{truthconv}}$ & maximum number of terms in Equation~\eqref{eq:true-coefficient} that the predictor sums & $4096$ \\
        $B_{\text{configweight}}$ & child blocks touched by an error configuration that \texttt{AddListed} lists & $2$ \\
        $B_{\text{class}}$ & sets of ambiguous classes kept per gap, needed only for $k > 1$ & $16$ \\
        \bottomrule
    \end{tabular}
\end{table}
Table~\ref{tab:budgets} lists the budgets of our predictor. We now describe what happens when they are exhausted. Algorithm~\ref{alg:dp} restricts prob to the budgets $B_{\text{gap}}$ and $B_{\text{cellratio}}$ each time it sets prob $\gets$ newProb, and Algorithm~\ref{alg:level} restricts each dist[$P$] to the budgets $B_{\text{gap}}$ and $B_{\text{likelihood}}$ before returning it. We build $\mathcal{J}$ once per code, from $\mathcal{L}$, by intersecting the members of $\mathcal{L}$ one at a time. We discard any intersection of weight below $w$, since intersecting it further only lowers its weight. When $|\mathcal{J}|$ would exceed $B_{\text{closure}}$, we instead sample $B_{\text{closure}}$ \emph{enclosed} error configurations. We do the same when the pairs $J' \subseteq J$ in $\mathcal{J}$, or the unions of the overlaps of one closure, would exceed fixed budgets. At level 1, we also sample $B_{\text{closure}}$ error configurations of cases~(b) and~(c) when there are more than $B_{\text{closure}}$ of them. Every coefficient then carries a sampling error, which we estimate at level 1. There are $\binom{n}{t} 3^t$ Pauli operators of weight $t$ on $n$ qubits, so the enumeration of stabilizers only reaches the largest weight $t$ with $\binom{n}{t} 3^t \le B_{\text{work}}$. By Lemma~\ref{lem:zero-difference}, every stabilizer $g$ of Equation~\eqref{eq:true-coefficient} has $\wt(g) \le 2d$, as we show in Appendix~\ref{app:predictor}. Thus, $B_{\text{work}}$ is exhausted only when $t < 2d$. Appendix~\ref{app:predictor} derives the runtime cost.
\section{Evaluation}\label{sec:evaluation}

\subsection{Methodology}\label{sec:methodology}
We evaluate our predictor on the codes of Table~\ref{tab:codes}. We concatenate each code with itself as long as the effective code distance is $< 30$ (except for iceberg, which is concatenated up to $5$ levels). Since the toric code has $k = 2$, we perform interleaved concatenation, as described in Section~\ref{sec:tree-decoder}. We use depolarizing noise in the code-capacity setting of Section~\ref{sec:ml-decoding}, and we report the LER of the first logical qubit. We use the default budgets of Table~\ref{tab:budgets}, except that above level 1 we set $B_{\text{work}} = 4.8 \times 10^{7}$. For a few surface codes at level 1, we also pick a higher $B_{\text{closure}}$ to ensure that the closures are enumerated within budget.

\begin{table}
    \centering
    \small
    \caption{The codes we evaluate in this work, the levels $\ell$ at which we predict their LERs, along with the reference baselines used.}
    \label{tab:codes}
    \begin{tabular}{l l l l}
        \toprule
        Code & $[[n, k, d]]$ & Levels & Baselines \\
        \midrule
        Iceberg~\cite{self_protecting_2024} & $[[4, 1, 2]]$ & $\ell \le 5$ & Monte Carlo for $\ell \le 5$, exact for $\ell \le 3$ \\
        Surface & $[[d^2, 1, d]]$, $d = 2\text{--}29$ & $\ell \le 4$ & for $d = 3$: Monte Carlo for $\ell \le 2$, exact for $\ell = 1$ \\
        Toric & $[[2d^2, 2, d]]$, $d = 2\text{--}29$ & $\ell \le 4$ & for $d = 2$: Monte Carlo for $\ell \le 2$, exact for $\ell = 1, 2$ \\
        Gross~\cite{bravyi_high-threshold_2024} & $[[144, 12, 12]]$ & $\ell = 1$ & --- \\
        GALA~\cite{yang_designer_2026} & $[[132, 30, 12]]$ & $\ell = 1$ & --- \\
        Mitten~\cite{bhardwaj_high-rate_2026} & $[[150, 30, 10]]$ & $\ell = 1$ & --- \\
        & $[[200, 40, 12]]$ & $\ell = 1$ & --- \\
        & $[[300, 60, 14]]$ & $\ell = 1$ & --- \\
        \bottomrule
    \end{tabular}
\end{table}

\begin{figure}
    \centering
    \includegraphics[width=\linewidth]{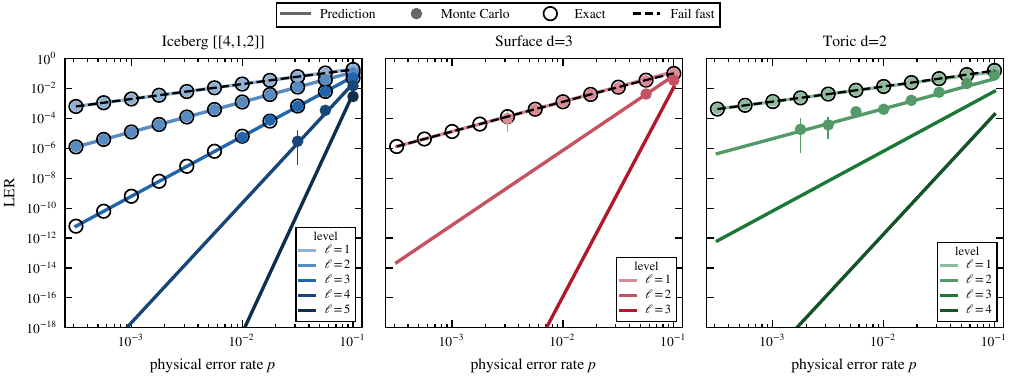}
    \caption{LER of the first logical qubit versus $p$ under depolarizing noise, with one panel per code. Lines show the prediction $c\,p^{o}$. Filled markers show Monte Carlo with 95\% confidence interval error bars. Open circles show exact enumeration, and dashed lines show fail fast at level 1.}
    \label{fig:validation}
\end{figure}

We compare our predictor with three baselines. First, we estimate the LER by Monte Carlo sampling with the tree decoder of Section~\ref{sec:tree-decoder}, with error bars showing 95\% confidence intervals. Second, we compute the exact LER of the same decoder by enumerating errors and modelling the action of the joint-marginal ML decoder on each error configuration. These two baselines are feasible only for small codes and distances as described in Table~\ref{tab:codes}. Third, at level 1, we compare with fail fast~\cite{beverland_fail_2025}. It estimates the coefficients of Equation~\eqref{eq:ler-expansion} from sampled error configurations of each weight, and fits an ansatz to them.

Figure~\ref{fig:validation} compares the predicted LER with these baselines. At level 1, the prediction agrees with exact enumeration to within $0.1\%$ at $p = 3.2 \times 10^{-4}$, validating Theorem~\ref{thm:exact}. Above level 1, the prediction agrees to within $0.2\%$ for the iceberg code up to level 3. For the toric code at level 2, the predicted leading coefficient equals the exact one.

\begin{figure}
    \centering
    \includegraphics[width=\linewidth]{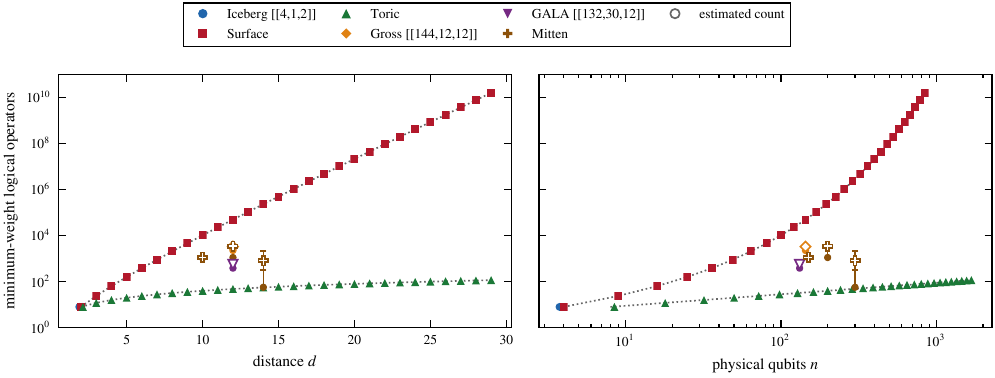}
    \caption{Number of minimum-weight logical operators for each code, against $d$ (left) and $n$ (right). Filled markers give the exact count, and dotted lines give $d\,2^{d}$ for the surface code and $4d$ for the toric code. For the high-rate codes, dots give the logical operators that the search finds, and open markers give the Chao1 estimate of their count, with its 95\% interval.}
    \label{fig:logicals}
\end{figure}

\subsubsection{Enumerating Logical Operators}

\begin{figure}
    \centering
    \includegraphics[width=\linewidth]{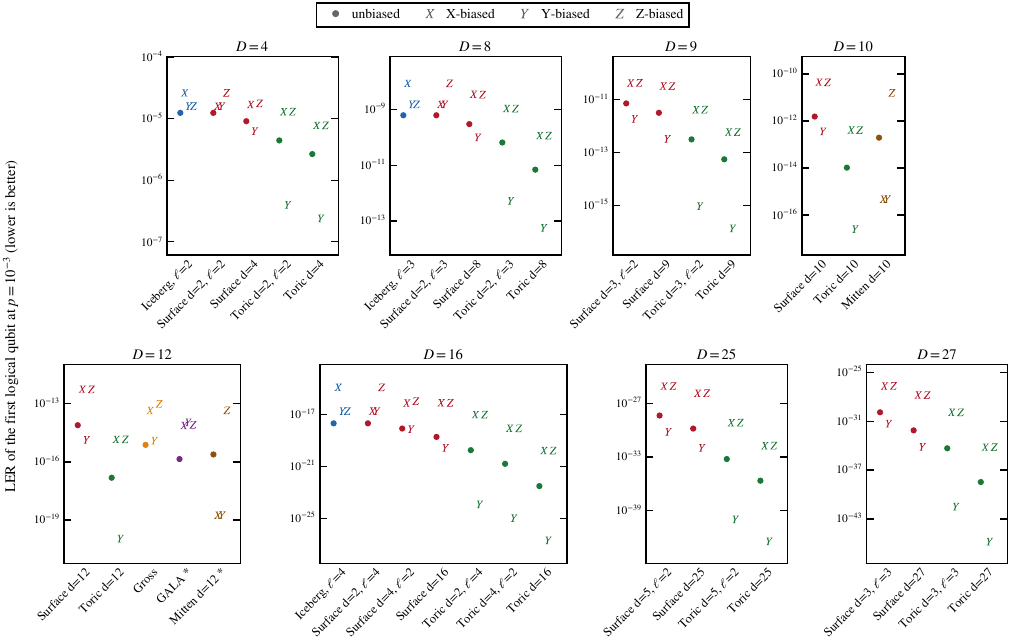}
    \caption{Predicted LER of the first logical qubit at $p = 10^{-3}$ under code-capacity noise, grouping codes with the same effective distance $D = d^{\ell}$. Dots give unbiased, i.e., depolarizing, noise, and the letters $X$, $Y$ and $Z$ give noise biased towards that Pauli, which occurs with probability $0.8p$, while each of the other two occurs with probability $0.1p$. An asterisk marks a high-rate code whose prediction uses only the logical operators that the search finds, as described in Section~\ref{sec:methodology}.}
    \label{fig:ranking}
\end{figure}
Our predictor takes $\mathcal{L}$ and, for odd $d$, $\mathcal{N}$ as inputs (Problem~\ref{prob:leading-ler}). We find them once per code, using a search that stops after $B_{\text{enum}} = 2 \times 10^{7}$ steps. If we do not find enough logical operators via a direct search, we make $B_{\text{search}} = 10^{4}$ calls to a decoder as described in~\cite{beverland_fail_2025}. A missing logical operator causes the predictor to miss some error configurations on which the actual decoder fails. Thus, if every budget of Table~\ref{tab:budgets} is large enough, the computed leading coefficient of the total LER at level 1 is a lower bound on the true one, as in~\cite{beverland_fail_2025}. This may not hold for the LER of any particular logical observable $\theta$, since a missing logical operator may also cause the predictor to miss a competing class that does not flip $\theta$, leading to a different failure fraction.

For the toric code, we list its $4d$ minimum-weight logical operators directly. For the surface code, $\mathcal{L}$ has $d\,2^{d}$ elements. For $d = 13$, we increase $B_{\text{enum}}$ to $10^{12}$, which suffices for the search. For $d \ge 14$, we instead count the size of $\mathcal{L}$ combinatorially, and sample $10^{4}$ elements of $\mathcal{L}$ uniformly at random instead of a search. For the high-rate codes, the search stops early, and we set $B_{\text{search}} = 10^{5}$, or $3 \times 10^{4}$ for the Mitten code with $d = 14$.

Figure~\ref{fig:logicals} shows the size of $\mathcal{L}$ for each code. For the high-rate codes, the search finds between $7\%$ and $95\%$ of $\mathcal{L}$, by the Chao1 estimate~\cite{chao_nonparametric_1984}. When the estimated chance that another call returns a logical operator already found is at least $0.94$, i.e., for the gross code and the Mitten code with $d = 10$, we scale the predicted leading coefficient by the ratio of the estimated size of $\mathcal{L}$ to the number of logical operators found. For the other high-rate codes, we report the prediction from the logical operators found, since the estimate is not confident enough to predict the total number of minimum-weight logical operators for the code.

\begin{figure}
    \centering
    \includegraphics[width=\linewidth]{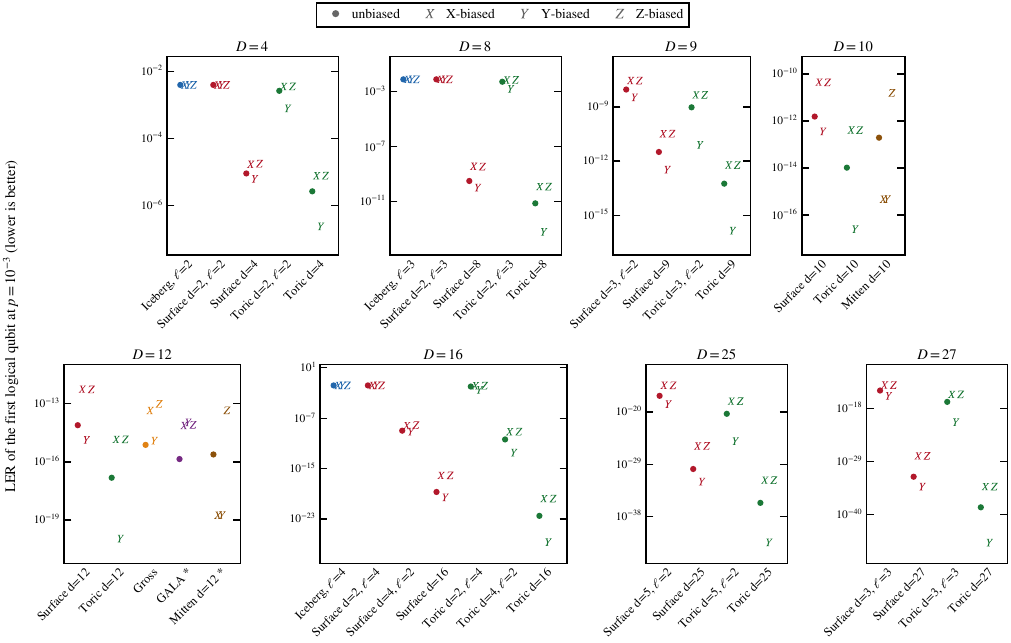}
    \caption{As in Figure~\ref{fig:ranking}, under phenomenological noise, where each level above the first also experiences logical measurement errors, at the LER of the level below. Codes at level 1 are therefore unchanged from Figure~\ref{fig:ranking}.}
    \label{fig:ranking-pheno}
\end{figure}

\subsection{Rank Ordering}

Using our predictor, we can now efficiently compare different code families to pick the code choice which satisfies desired properties, for example, encoding rate at least $e^*$, or number of logical qubits $k \ge k^*$, etc. for some choice of $e^*, k^*$, and still has the best logical error rate. Evaluating this without our predictor would require performing Monte Carlo sampling on each code, requiring days, or even years of computation (Figure~\ref{fig:runtime}). We demonstrate applications of our predictor in Figures~\ref{fig:ranking} and~\ref{fig:ranking-pheno}, where we compare different codes under the code-capacity and phenomenological noise models, respectively. Each code is grouped by $D = d^\ell$, the effective distance of the code. We also compare different codes in their \emph{protection efficiency}, and \emph{encoding efficiency} in Figure~\ref{fig:ranking-frontier} which lets one pick the optimal code at a target LER, that minimizes the number of physical qubits, $n$, or the number of physical qubits per logical qubit, $n/k$, respectively.
\par Figure~\ref{fig:ranking} shows that the toric code has the lowest LER at every effective distance, under unbiased, $Y$-biased and $Z$-biased noise. Every configuration of the iceberg, surface and toric codes with effective distance $D$ uses $D^2$ physical qubits per logical qubit, so each group compares them at equal cost. In every group, the code at level 1 has a lower LER than the same family concatenated to the same $D$. The best code also depends on the noise bias: under $X$-biased noise, the Mitten codes have the lowest LER at $D = 10$ and $D = 12$, about $900$ and $8000$ times lower than the toric code, while under $Z$-biased noise their LER is $30$ to $40$ times higher than the toric code's. Under phenomenological noise (Figure~\ref{fig:ranking-pheno}), every concatenated code gets worse, since each level adds logical measurement errors at the LER of the level below, while codes at level 1 are unchanged.
\par Figure~\ref{fig:ranking-frontier} shows that GALA and the Mitten codes need the fewest physical qubits per logical qubit at every target LER. In physical qubits, however, the surface code is the cheapest for every target LER down to about $2 \times 10^{-14}$ at $p = 10^{-3}$, and for nearly every target LER down to $10^{-15}$ at $p = 10^{-4}$. Below that target at $p = 10^{-3}$, GALA needs fewer physical qubits than any surface code, and so does the gross code below about $8 \times 10^{-15}$. For example, at $p = 10^{-3}$ and a target LER of $10^{-12}$, the surface code with $d = 11$ needs the fewest physical qubits, $121$, but for $e^* = 0.2$, GALA reaches the target with $132$ physical qubits for $30$ logical qubits.
\par Our predictor produced every point in Figures~\ref{fig:ranking} to~\ref{fig:ranking-frontier} in $48$ core-hours in total, and the same runs give the predicted LER at every $p$. Monte Carlo, in contrast, must be rerun for each $p$, and at $p = 10^{-3}$ alone it would need more than $10^{12}$ shots for $113$ of the $148$ points in Figure~\ref{fig:ranking}.

\subsection{Sensitivity to Budgets}
Figure~\ref{fig:sensitivity} varies each budget of Table~\ref{tab:budgets} around its default, one at a time, for six concatenated codes at their deepest level and for the surface codes with $d = 9$ and $d = 11$ at level 1. Four budgets, $B_{\text{cellratio}}$, $B_{\text{class}}$, $B_{\text{configweight}}$ and $B_{\text{truthconv}}$, never change the predicted LER. Lowering $B_{\text{gap}}$ or $B_{\text{likelihood}}$ below its default changes the predicted LER, by up to $34\%$ at half the default $B_{\text{gap}}$, and by up to nine orders of magnitude at $B_{\text{gap}} = 2$. Above their defaults, neither changes the LER by more than $3\%$, while $B_{\text{likelihood}}$ multiplies the run time by up to $12$ at $16$ times its default. At the default budgets, every prediction is within $3\%$ of its value at the largest budget, except for the surface code with $d = 4$ at level 2, whose LER is $15\%$ higher at the default $B_{\text{work}}$, and stops changing from twice the default on. We therefore set $B_{\text{work}} = 4.8 \times 10^{7}$ above level 1. $B_{\text{closure}}$ changes only the surface codes with $d = 9$ and $d = 11$ at level 1, whose error configurations it samples: at its default, their LER is within $2.2\%$ of its value at the largest budget, and their run time grows about linearly with $B_{\text{closure}}$.

\begin{figure}
    \centering
    \includegraphics[width=\linewidth]{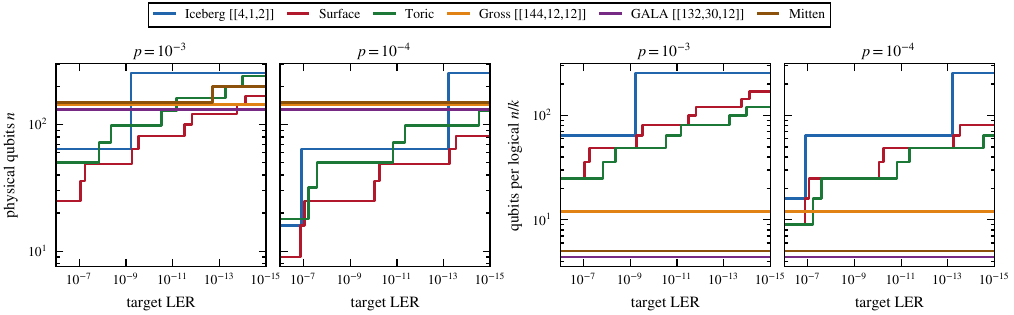}
    \caption{Plots depicting the number of physical qubits $n$ (left two panels) and physical qubits per logical qubit $n/k$ (right two panels) required by each code family to achieve a target LER, at $p = 10^{-3}$ and $p = 10^{-4}$. Each curve takes the cheapest $d$ and $\ell$ in its family whose predicted LER of the first logical qubit, under depolarizing code-capacity noise, is at most the target.}
    \label{fig:ranking-frontier}
\end{figure}

\begin{figure}
    \centering
    \includegraphics[width=\linewidth]{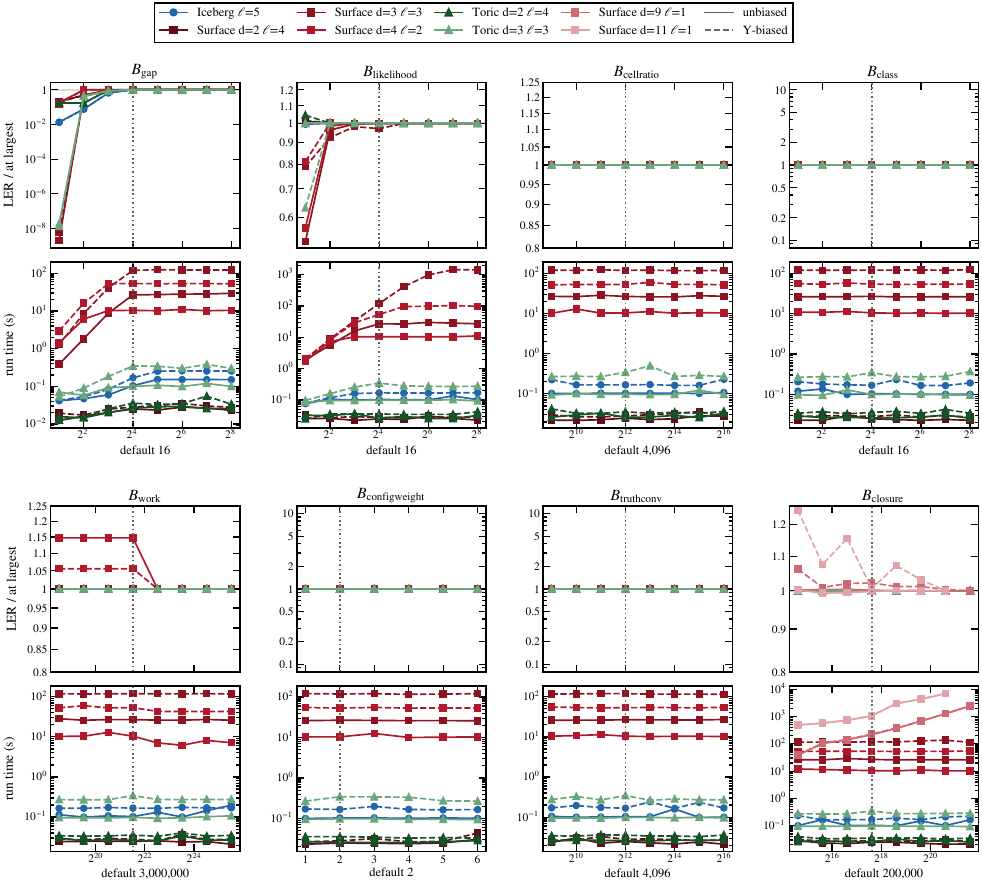}
    \caption{Sensitivity of our predictor to each budget of Table~\ref{tab:budgets}, varied one at a time around its default (dotted line), with the other budgets at their defaults, under code-capacity noise. For each budget, the top panel shows the predicted LER of the first logical qubit at $p = 10^{-3}$, relative to its value at the largest budget we ran, and the bottom panel shows the run time. Each line is one code at its deepest level, under unbiased (solid) or $Y$-biased (dashed) noise, and the surface codes with $d = 9$ and $d = 11$ at level 1 appear only for $B_{\text{closure}}$.}
    \label{fig:sensitivity}
\end{figure}
\section{Conclusion}\label{sec:conclusion}
In this work, we propose a predictor that computes the leading-order approximation $c\,p^{o}$ of the LER of a QEC code at every level of concatenation, taking the stabilizers, the minimum-weight logical operators and, for odd $d$, the logical operators of weight $d + 1$ of the code as input. Our predictions are exact for any unconcatenated stabilizer code whenever no budget is exhausted. For concatenated codes, we first prove a recursion that computes the leading-order approximation exactly, for an arbitrary number of logical qubits per block under interleaved concatenation. We then approximate this recursion with an algorithm that, for fixed budgets, runs in time polynomial in $d$ and in the number of minimum-weight logical operators, and linear in the number of levels. Our predictions match exact enumeration at level 1 on every code we evaluate and up to level 3 for the iceberg code. For the toric code at level 2, our predicted leading coefficient equals the exact leading coefficient. Our predictor thus enables an efficient search among code families for the code with the lowest LER that has the desired properties, such as a minimum encoding rate or a minimum number of logical qubits. The same search would be intractable with Monte Carlo sampling, which needs more than $10^{12}$ shots at $p = 10^{-3}$ for $113$ of the $148$ cases we compare.

\section*{Acknowledgements}
The authors acknowledge the Texas Advanced Computing Center (TACC) at The University of Texas at Austin for providing computational resources that have contributed to the research results reported in this paper. This work was funded in part by the Texas Quantum Institute (TQI). AI tools were used to assist with implementing parts of the evaluation framework, and for help with drafting some technical details of the theorems. The authors developed the key ideas in the paper, and verified all AI-assisted contributions.
\appendix
\section{Proofs for Unconcatenated Codes}\label{app:unconcatenated}
For Pauli operators $A$ and $B$, let $\alpha(A, B)$ be the number of qubits where $A$ and $B$ carry the same non-identity Pauli, and $\beta(A, B)$ the number of qubits where both act but with different Paulis. Counting qubit by qubit,
\begin{equation}\label{eq:product-weight}
\wt(A \cdot B) = \wt(A) + \wt(B) - 2\alpha(A, B) - \beta(A, B),
\end{equation}
so $\wt(A \cdot B) \le \wt(A) + \wt(B)$. The proofs below also use the \emph{distance argument}: a Pauli operator with trivial syndrome and a non-trivial class is a non-trivial logical operator, so it has weight at least $d$. In particular, two operators with the same syndrome but different classes have weights that sum to at least $d$.

\begin{proof}[Proof of Lemma~\ref{lem:leading-terms}]
Each part follows from the definition of $\sim$ by taking the ratio of the two sides as $p \to 0$. For the sums and maxima in (i), this ratio lies between the smallest and the largest of the ratios $f_a/g_a$, which all tend to $1$.
\end{proof}

\begin{proof}[Proof of Lemma~\ref{lem:decision-rule}]
For every class $\lambda \neq \lambda^*$, Equation~\eqref{eq:leading-limit} gives $W_\lambda(s) / W_{\lambda^*}(s) \sim \left( m_\lambda(s) / m_{\lambda^*}(s) \right) p^{\,\omega_\lambda(s) - \omega_{\lambda^*}(s)}$, which tends to $0$ or to a value below $1$. Each ratio is therefore below $1$ for small enough $p$, and $p_0$ is the smallest of these finitely many thresholds.
\end{proof}

\begin{proof}[Proof of Corollary~\ref{cor:gap-fraction}]
The classes of smallest leading order and, among them, largest leading coefficient are $\lambda_0$ alone, $\lambda_0$ together with $\Lambda_E$, or $\Lambda_E$ alone, in the three cases of Equation~\eqref{eq:gap-fraction}. By Lemma~\ref{lem:decision-rule}, the decoder returns one of them, chosen evenly when several are ambiguous. Returning $\lambda_0$ never fails, and returning $\lambda \in \Lambda_E$ fails for $\theta$ exactly when $\lambda \in \Lambda_E^\theta$.
\end{proof}

\begin{proof}[Proof of Lemma~\ref{lem:min-weight}]
Let $E$ have weight below $\lceil d/2 \rceil$, class $\lambda_0$ and syndrome $s$. By the distance argument, every error configuration with syndrome $s$ in another class has weight at least $d - \wt(E) > \wt(E) \ge \omega_{\lambda_0}(s)$. Thus $\gap(E) > 0$, the decoder corrects $E$ by Lemma~\ref{lem:decision-rule}, and $o_\theta \ge \lceil d/2 \rceil$.
\par For the equality, let $L$ be a logical operator of weight $d$ that flips $\theta$, take $E \subseteq L$ of weight $\lceil d/2 \rceil$, and let $E' = E \cdot L$, of weight $\lfloor d/2 \rfloor$. Since $E$ and $E'$ have the same syndrome, the decoder returns the same class for both, and since their classes differ by $\mathsf{cls}(L)$, exactly one of the two residuals flips $\theta$, so $f_\theta(E) + f_\theta(E') = 1$. For odd $d$, $f_\theta(E') = 0$ by the first part, so $f_\theta(E) = 1$. For even $d$, both have weight $d/2$, so one of them has $f_\theta > 0$. Either way, $F_{\theta, \lceil d/2 \rceil} > 0$.
\end{proof}

\begin{proof}[Proof of Lemma~\ref{lem:competitors}]
For (i), let $E' = E \cdot M$, with $\alpha = \alpha(E, E')$ and $\beta = \beta(E, E')$. Equation~\eqref{eq:product-weight} gives $\wt(M) = w + \wt(E') - 2\alpha - \beta$, and $\wt(M) \ge d$ with $\wt(E') \le w$ gives $2\alpha + \beta \le 2w - d$, which is $0$ for even $d$ and $1$ for odd $d$. If $2\alpha + \beta = 0$, then $E \subseteq M$, and the weights give case~(a), or, for odd $d$ and $\wt(M) = d + 1$, case~(c). If $2\alpha + \beta = 1$, then $d$ is odd, and $E$ and $E'$ share one qubit with different Paulis, on which $M$ differs from $E$, and the weights give $\wt(M) = d$, which is case~(b). In every case $\supp(E) \subseteq \supp(M)$.
\par For (ii), let $E'$ have syndrome $s$, a competing class and $\wt(E') \le w$. By the distance argument, $M = E \cdot E'$ is a non-trivial logical operator, and $d \le \wt(M) \le 2w \le d + 1$, so $M \in \mathcal{C}_E$. The converse is immediate. Part (iii) follows from (ii) by keeping the lightest $E \cdot M$. For (iv), if $\mathcal{C}_E$ is empty, then by (ii) every competing error configuration at $s$ has weight above $w \ge \omega_{\lambda_0}(s)$, so $\gap(E) > 0$ and $f_\theta(E) = 0$.
\end{proof}

\begin{proof}[Proof of Lemma~\ref{lem:true-class}]
The error configurations with syndrome $s$ and class $\lambda_0$ are the $E \cdot g$ for $g \in S$. If $\wt(E \cdot g) \le w$, Equation~\eqref{eq:product-weight} gives $\wt(g) \le \wt(E) + \wt(E \cdot g) \le 2w \le d + 1$, and $\omega_{\lambda_0}(s) \le w$, since $E$ itself has weight $w$.
\end{proof}

\begin{proof}[Proof of Lemma~\ref{lem:enclosed-classes}]
For (i), each $M \in \mathcal{L}_E$ has $\wt(E \cdot M) = \lfloor d/2 \rfloor$, the smallest weight among the cases of Lemma~\ref{lem:competitors}~(i), so Lemma~\ref{lem:competitors}~(iii) gives $\omega_{\neg\lambda_0}(s)$ and $\Lambda'_E$. Any error configuration $F$ with syndrome $s$ and class $\lambda_0$ lies in a different class from $E \cdot M$, so the distance argument gives $\wt(F) \ge d - \lfloor d/2 \rfloor = w$, i.e., $\omega_{\lambda_0}(s) = w$.
\par For (ii), let $d$ be odd and $M, M' \in \mathcal{L}_E$. Up to a sign, $M \cdot M' = (E \cdot M) \cdot (E \cdot M')$ has trivial syndrome and weight at most $d - 1$, so by the distance argument its class is trivial, and $\mathsf{cls}(M) = \mathsf{cls}(M')$. Thus, $\Lambda_E$ is the single class $\lambda_0 \cdot \mathsf{cls}(M)$, $\gap(E) = -1$ by (i), and Corollary~\ref{cor:gap-fraction} gives $f_\theta(E) = 1$ exactly when $M$ flips $\theta$.
\end{proof}

\begin{proof}[Proof of Lemma~\ref{lem:mobius-count}]
Let $J = \bigcap_{M \in \mathcal{L}'} M$ be in $\mathcal{J}$. Every error configuration $E \subseteq J$ of weight $w$ is \emph{enclosed}, and an \emph{enclosed} $E$ has $E \subseteq J$ exactly when $J_E \subseteq J$, since $E \subseteq J_E$, and $E \subseteq J$ gives $\mathcal{L}' \subseteq \mathcal{L}_E$. Equation~\eqref{eq:mobius} with $K = J_E$ then gives $[J_E = J] = \sum_{J' \in \mathcal{J},\ E \subseteq J' \subseteq J} \mu(J', J)$ for $E \subseteq J$, so,
\begin{align}
\sum_{E \in \{E:\ E\text{ is \emph{enclosed}}\}} \varphi(J_E, E) &= \sum_{J \in \mathcal{J}} \; \sum_{\substack{E \subseteq J \\ \wt(E) = w}} [J_E = J]\, \varphi(J, E) \nonumber \\
&= \sum_{J \in \mathcal{J}} \; \sum_{\substack{J' \in \mathcal{J} \\ J' \subseteq J}} \mu(J', J) \left( \sum_{\substack{E \subseteq J' \\ \wt(E) = w}} \varphi(J, E) \right), \nonumber
\end{align}
where the first step counts each \emph{enclosed} $E$ at its closure, and the second substitutes $[J_E = J]$ and exchanges the sums over $E$ and $J'$.
\end{proof}

\section{Proofs for Concatenated Codes}\label{app:concatenated}
\begin{proof}[Proof of Lemma~\ref{lem:parent-soft}]
Dividing Equation~\eqref{eq:parent-likelihood} for $P$ by the same equation for $P = I$, and both by $\prod_j W^{(j)}_{\lambda_j}(s^{(j)})$, turns each child block $j \in \supp(M)$ into the ratio $W^{(j)}_{\lambda_j \cdot \mathsf{cls}(\bar{M_j})}(s^{(j)}) / W^{(j)}_{\lambda_j}(s^{(j)})$, and every other child block into $1$. By Equation~\eqref{eq:soft-ratio}, the left side has leading term $\rho_P(E)\, p^{\gap_P(E)}$, and each ratio has leading term $\rho_{M_j}(E^{(j)})\, p^{\gap_{M_j}(E^{(j)})}$. Lemma~\ref{lem:leading-terms} then gives Equation~\eqref{eq:parent-soft}, whose right side depends on $E$ only through $\sigma(E^{(1)}), \ldots, \sigma(E^{(n)})$.
\end{proof}

\begin{proof}[Proof of Theorem~\ref{thm:recursion}]
Since the child blocks act on disjoint qubits, $\Pr(E) = \prod_j \Pr(E^{(j)})$, and by Lemma~\ref{lem:parent-soft}, $\sigma(E) = \Phi(\sigma(E^{(1)}), \ldots, \sigma(E^{(n)}))$. Grouping each $E^{(j)}$ by its soft information $\tau_j$ therefore factors the sum of $\Pr(E)$ over $\sigma(E) = \tau$ into the first part of Equation~\eqref{eq:recursion}. Replacing each probability by its leading term from Equation~\eqref{eq:soft-leading} gives the second part, by Lemma~\ref{lem:leading-terms}~(i), and $\Phi$ depends only on the code of the parent block.
\end{proof}

\section{Proofs for the Predictor}\label{app:predictor}
\begin{proof}[Proof of Lemma~\ref{lem:zero-difference}]
Let $c = |\supp(g) \setminus \supp(M)|$. Equation~\eqref{eq:product-weight} gives $\wt(M \cdot g) = d + c - \alpha(M, g)$, and $M \cdot g$ is a non-trivial logical operator, so $c \ge \alpha(M, g)$ by the distance argument. Counting qubit by qubit, with $\supp(F) \subseteq \supp(M)$,
\begin{equation*}
\wt(F \cdot g) - \wt(F) = c + \left| \left( \supp(g) \cap \supp(M) \right) \setminus \supp(F) \right| - \alpha(F, g) \ge c - \alpha(M, g) \ge 0,
\end{equation*}
since $\alpha(F, g) \le \alpha(M, g)$, as $F$ acts as $M$ on $\supp(F)$. Equality holds exactly when $\supp(g) \cap \supp(M) \subseteq \supp(F)$, which also gives $\alpha(F, g) = \alpha(M, g)$, and $c = \alpha(M, g)$, i.e., $\wt(M \cdot g) = d$.
\end{proof}

\paragraph{Weight of the stabilizers in Equation~\eqref{eq:true-coefficient}.}
By Lemma~\ref{lem:zero-difference}, every stabilizer $g$ of Equation~\eqref{eq:true-coefficient} has $\wt(M \cdot g) = d$ and $\supp(g) \cap \supp(M) \subseteq \supp(\tilde E)$ for every $M \in \mathcal{L}_J$. Equation~\eqref{eq:product-weight} with $\wt(M \cdot g) = \wt(M)$ then gives $\wt(g) = 2\alpha(M, g) + \beta(M, g) \le 2\,|\supp(g) \cap \supp(M)| \le 2\wt(\tilde E) \le 2d$.

\paragraph{Runtime cost.}
The construction of $\mathcal{J}$ and the enumeration of stabilizers happen once per code, within their budgets. At each level, Algorithm~\ref{alg:dp} runs once for each pair $J' \subseteq J$ in $\mathcal{J}$. Since \texttt{state.numTouched} and \texttt{state.numNotLightest} are at most $d$, a call keeps at most $(d + 1)^2 (2B_{\text{gap}} + 1) B_{\text{cellratio}}\times F$ states, where $F$ is the number of values of \texttt{state.fullyTouched}, which is bounded by a fixed budget on the unions of overlaps. In each of at most $d$ sets of \emph{identical} child blocks, each state has at most $\binom{d + |\mathcal{D}_P| - 1}{d}$ splits, where $|\mathcal{D}_P| \le 2B_{\text{gap}} B_{\text{likelihood}} + 1$, and each split takes $O(d)$ time. A call to Algorithm~\ref{alg:dp} therefore takes time,
\begin{equation}\label{eq:call-time}
O\left( d^2 (d + 1)^2 (2B_{\text{gap}} + 1) B_{\text{cellratio}} \binom{d + |\mathcal{D}_P| - 1}{d} \right),
\end{equation}
which is polynomial in $d$ for fixed budgets, since the binomial coefficient has degree $|\mathcal{D}_P| - 1$ in $d$. \texttt{AddListed} adds $O(|\mathcal{L}|\, d^{B_{\text{configweight}}})$ error configurations per level, each with at most $B_{\text{truthconv}}$ terms of Equation~\eqref{eq:true-coefficient}. None of these bounds depends on $\ell$, so the runtime is linear in $\ell$.

\bibliographystyle{alpha}
\bibliography{references}

\end{document}